\documentclass[a4paper,twocolumn,11pt,unpublished]{quantumarticle}

\pdfoutput=1
\usepackage{amsmath,amssymb,amsthm,mathtools}
\usepackage{physics}
\usepackage{hyperref}
\usepackage{cleveref}
\usepackage{braket}
\usepackage{graphicx}
\usepackage{subcaption}
\usepackage{tikz}
\usetikzlibrary{arrows.meta,positioning,fit}
\numberwithin{equation}{section}

\makeatletter
\long\def\addtocontents#1#2{%
  \protected@write\@auxout
    {\robustify@contents\let\glossary\@gobble}%
    {\string\@writefile{#1}{#2}}%
}
\makeatother
\usepackage[
  backend=bibtex,
  style=alphabetic,
  doi=true,
  url=true,
  eprint=false,
]{biblatex}
\AtEveryBibitem{\clearfield{note}}
\newcommand{\poly}{\text{poly}}

\newtheorem{theorem}{Theorem}
\newtheorem{definition}{Definition}
\newtheorem{proposition}{Proposition}

\newtheorem{lemma}{Lemma}

\title{Bosonic Encodings for Hermite-Galerkin Discretizations of High-Dimensional PDEs and Bayesian Inverse Problems}
\author{Alice Barthe}
\affiliation{PsiQuantum, 700 Hansen Way, Palo Alto, CA 94304, USA}
\date{August 2026}

\begin{document}

\begin{abstract}
The Koopman-von Neumann framework has been proposed to design quantum algorithms for non-linear dynamics. 
It maps a non-linear ordinary differential equation to a linear partial differential equation (PDE) governing a probability amplitude. 
Previous works represents this amplitude in the Hermite-function basis, equivalently as a bosonic state, and truncates the total Hermite degree to obtain a representation over $\Theta(m\log N)$ qubits, where $N$ is the number of variables and $m$ the truncation order. 
We extend this approach to a broader class of linear PDEs whose differential operators have a structured polynomial form. 
We prove convergence of the truncation for both time-dependent dynamics and gapped ground-state problems under explicit regularity and stability assumptions. 
We then introduce a qubit encoding that supports efficient block encodings of the truncated operators. 
Finally, we apply the framework to Bayesian inverse problems with Gaussian priors and observation noise, reducing posterior-state preparation to the preparation of a structured Hamiltonian's ground state.
\end{abstract}

\maketitle

\section{Introduction}

Quantum algorithms offer a promising route toward accelerating the numerical solution of differential equations, with theoretical speed-ups established for certain classes of linear ordinary differential equations (ODEs) and partial differential equations (PDEs) under suitable sparsity, conditioning, and state-preparation assumptions~\cite{berry2017quantum-223}.
In this work, we consider linear PDEs in a large number of variables. 
Although physical space is typically low-dimensional, a PDE may instead describe a function of the full state of a system. 
For example, the Liouville equation describes the evolution of a probability density under deterministic dynamics, while the Fokker--Planck equation additionally accounts for stochastic noise. 
For a system with many degrees of freedom, these probability densities depend on a correspondingly large number of variables.

Our first contribution, presented in \Cref{sec:global-hermite-galerkin}, is to extend an existing framework developed for specific PDEs such as the Liouville equation~\cite{engel2021linear-e8c,tanaka2023polynomial-272} and the backward Kolmogorov equation~\cite{bravyi2025quantum-101} to a broader family of linear PDEs.
Hermite functions are the eigenfunctions of the quantum harmonic oscillator, and they can equivalently be interpreted as bosonic occupation states~\cite{kowalski1997nonlinear-d40}. 
We expand solutions to the PDE in this infinite basis, and truncate by retaining only terms whose total Hermite degree is at most $m$, an approach known as global truncation as proposed in~\cite{engel2021linear-e8c}. 
This truncation strategy, further explained in \Cref{subsec:background}, is key to enabling quantum advantage, with a state that can be represented using $O(m\log N)$ qubits.
Under explicit regularity and stability assumptions, we prove convergence of this truncation for both time-dependent evolution and steady-state solutions, which in our framework correspond to static ground-state problems.

Our second contribution, presented in \Cref{sec:qubit-encoding}, is a qubit encoding of the globally truncated space. 
This space can be viewed either as bosonic states containing at most $m$ bosons distributed among $N$ modes, or as functions of $N$ variables expanded in the Hermite basis with total Hermite degree at most $m$. 
We encode it using $m$ registers, each storing either a mode label $j\in\{1,\ldots,N\}$ or an empty symbol indicating that fewer than $m$ bosons are present. 
Because bosons are indistinguishable, the ordering of the registers cannot carry physical information, so the encoded state is invariant under permutations of the registers. 
Each register requires $\lceil\log_2(N+1)\rceil$ qubits. 
Using this representation, we construct modular block encodings of bosonic creation and annihilation operators and combine them to encode structured polynomial differential operators with gate complexity polynomial in $m$ and $\log N$.

Our final contribution, presented in \Cref{sec:fokker-planck}, is to combine the standard transformation of a reversible Fokker--Planck operator into a Hermitian ground-state problem~\cite[Sec.~6.3]{risken1996fokker-planck-aa7} with the Hermite truncation and qubit encoding developed in this work. 
For Bayesian inverse problems with Gaussian priors and Gaussian observation noise, whitening the parameters yields a $\|x\|^2_2$ term in the posterior potential. 
The corresponding Hamiltonian can then be written as the bosonic number operator plus a data-dependent perturbation. 
We give sufficient conditions on the forward model under which this perturbation has a structured polynomial form and can be block encoded efficiently. 
We then derive the resulting ground-state-preparation complexity in terms of the truncation error, block-encoding normalization, spectral gap, and initial-state overlap, and discuss when these requirements lead to a tractable quantum algorithm.

\section{Global Hermite Truncation: Setup and Convergence}
\label{sec:global-hermite-galerkin}

In this section, we recall how Hermite expansions have been used to obtain finite-dimensional approximations of the linear PDEs arising from non-linear dynamical systems. 
We extend this approach to a broader class of high-dimensional linear PDEs. 
Our contribution in this section is to prove convergence of the global Hermite truncation for both time-dependent evolution and ground-state problems under explicit regularity and stability assumptions, and to quantify the dependence of the error on the truncation order.

\subsection{Relation to prior work}
\label{subsec:prior}

The approach considered here originates from methods that replace non-linear ordinary differential equations by linear equations on a larger function space. 
Consider the non-linear system
\begin{equation}
\dot{x}=f(x),
\qquad
x\in\mathbb R^N.
\end{equation}
A family of such linearisation methods was introduced in~\cite{kowalski1997nonlinear-d40}. 
This family includes the Koopman--von Neumann framework, on which the present work builds.

Rather than following a single trajectory $x(t)$, the Koopman--von Neumann framework considers a probability distribution over possible states. 
If every point evolves according to the original ODE, its probability density $\rho(t,x)$ evolves according to the linear Liouville equation. 
Equivalently, the probability amplitude
\begin{equation}
\psi(t,x):=\sqrt{\rho(t,x)}
\end{equation}
satisfies a linear evolution equation generated by the Koopman--von Neumann Hamiltonian. 
Although the original system has only $N$ variables, this linear equation acts on an infinite-dimensional space of functions of those variables.

A quantum algorithm therefore requires a finite-dimensional approximation. 
One possibility is to discretize each variable on a grid and represent the solution by its values on that grid, as proposed in~\cite{joseph2020koopmanvon-960}. 
This approach comes with convergence guarantees and block encodings, but its qubit requirement is linear in $N$ and therefore does not provide an asymptotic space advantage over classical representations.

An alternative proposed in~\cite{engel2021linear-e8c} is to expand the solution in Hermite functions and retain only basis functions whose total Hermite degree is at most $m$. 
We call this approach global truncation and define it precisely in \Cref{subsec:background}. 
The resulting space can be represented using $O(m\log N)$ qubits. The analysis of~\cite{engel2021linear-e8c} shows improved convergence as the strength of the non-linearity decreases, but does not give a general error bound for the stronger non-linearities considered here. 
Later work~\cite{tanaka2023polynomial-272} refines the quantum-algorithm complexity, but does not fully establish a global-truncation error bound in the regimes studied in this work.

Bravyi et al.~\cite{bravyi2025quantum-101} and subsequent work~\cite{bravyi2026quantum-0f7} extend this line of research to noisy non-linear ODEs. For a restricted family of equations, they prove both convergence and efficient block encoding. Their convergence mechanism relies on noise and a confining linear term that suppress components of high Hermite degree. 
Among the preceding works, the present paper is most directly inspired by Bravyi et al., whose use of global Hermite truncation and noise-induced suppression of high Hermite degrees provides the closest conceptual starting point for our analysis.
We examine the precise relationship between their setting and ours in \Cref{subsec:general-fokker-planck}.

Finally, recent work~\cite{wu2026from-49e} introduces a mapping from non-linear stochastic differential equations to Lindbladian dynamics. 
For each Brownian realization, they represent the corresponding conditional density by an $L^2$-normalized half-density whose stochastic unitary evolution, after averaging the associated projectors, yields a Lindblad equation. 
The diagonal of the resulting density operator reproduces the Fokker--Planck density. 

\subsection{Background: Hermite--Galerkin expansion}
\label{subsec:background}
In this subsection, we recall the existing concepts we need to deploy our framework. 
Let $\mathcal{L}$ be a linear differential operator on $\mathbb R^N$. 
We consider the linear PDE
\begin{equation}
\label{eq:linear-pde-general}
    \partial_t u(t,x)=\mathcal{L}u(t,x), \quad x \in \mathbb{R}^N.
\end{equation}
We expand the solution $u$ on the Hermite function basis $h_\alpha$, which is also called a Galerkin expansion, as follows
$$
    u(t,x)=\sum_\alpha c_\alpha(t)h_\alpha(x).
$$
Let $\{\text{He}_k\}_{k\geq 0}$ be the one-dimensional Hermite polynomials in the physicists' convention.
The normalized one-dimensional Hermite functions are
$$
h_k(x)
=
\frac{1}{\sqrt{2^k k!\sqrt{\pi}}}\text{He}_k(x)e^{-x^2/2}.
$$
For a multi-index $\alpha=(\alpha_1,\ldots,\alpha_N)\in\mathbb N^N$, the $N$-dimensional Hermite function is the product:
\begin{equation}
\label{eq:hermite-functions-definition}
    h_\alpha(x)
    =
    \prod_{j=1}^N h_{\alpha_j}(x_j).
\end{equation}
The family $(h_\alpha)_{\alpha\in\mathbb N^N}$ is an orthonormal basis of square integrable functions on unbounded domain $L^2(\mathbb R^N,dx)$.
We project the dynamics onto the basis which transforms the linear PDE into an infinite-dimensional linear ODE
\begin{equation}
\label{eq:coefficient-system-general}
\dot c_\alpha(t)
=
\sum_\beta L_{\alpha\beta}c_\beta(t),
\quad
L_{\alpha\beta}
:=
\langle h_\alpha,Lh_\beta\rangle_{L^2(\mathbb R^N)}.
\end{equation} 

Hermite functions are the eigenfunctions of the quantum harmonic oscillator. 
In one dimension, $h_k$ is the eigenfunction with excitation number $k$. 
In $N$ dimensions, the product function $h_\alpha$ describes $N$ independent oscillators with excitation numbers $\alpha_1,\ldots,\alpha_N$. 
This basis can equivalently be written as the occupation-number basis, also called the Fock basis
\begin{equation}
\ket{\alpha}
=
\ket{\alpha_1,\ldots,\alpha_N},\quad\alpha\in \mathbb{N}^N
\end{equation}
where $\alpha_j$ is the number of excitations in mode $j$~\cite{kowalski1997nonlinear-d40}.

In this representation, multiplication by the coordinate $x_j$ and differentiation with respect to $x_j$ correspond to the position and momentum operators of mode $j$:
\begin{equation}
\hat q_j=x_j,
\qquad
\hat p_j=-i\partial_{x_j}.
\end{equation}
More details can be found in the continuous-variable quantum-computing literature~\cite{lloyd1999quantum-d23,buck2021continuous-e94}.
The corresponding annihilation and creation operators are
\begin{equation}
\label{eq:ladder-operators}
\hat a_j
:=
\frac{1}{\sqrt 2}(\hat q_j+i\hat p_j),
\qquad
\hat a_j^\dagger
:=
\frac{1}{\sqrt 2}(\hat q_j-i\hat p_j).
\end{equation}
They respectively decrease and increase the excitation number of mode $j$:
\begin{equation}
\label{eq:ladder-action-hermite}
\hat a_j h_\alpha
=
\sqrt{\alpha_j}\,h_{\alpha-e_j},
\qquad
\hat a_j^\dagger h_\alpha
=
\sqrt{\alpha_j+1}\,h_{\alpha+e_j},
\end{equation}
where $e_j$ has value one in position $j$ and zero elsewhere.
The total Hermite degree
\begin{equation}
|\alpha|:=\sum_{j=1}^N\alpha_j
\end{equation}
is therefore the total number of excitations. 
It is measured by the number operators
\begin{equation}
\hat n_j:=\hat a_j^\dagger\hat a_j,
\qquad
\hat n:=\sum_{j=1}^N\hat n_j.
\end{equation}
With the conventions above,
\begin{equation}
\label{eq:number-operator-oscillator}
2\hat n
=
-\Delta+\|\hat q\|^2-N.
\end{equation}

Consequently, a polynomial differential operator can be written using position and momentum operators:
\begin{align}
    \mathcal{L} = &P(x_1,\cdots,x_N,\partial_{x_1},\cdots,\partial_{x_N})\\
    \label{eq:L(q,p)}
    \leftrightarrow &P(\hat q_1,\cdots,\hat q_N,i\hat p_1,\cdots,i\hat p_N)
\end{align}

The system in \Cref{eq:coefficient-system-general} is infinite-dimensional. 
In order to obtain a finite-dimensional problem, we must choose a finite set of states onto which to project. 
The approach chosen by~\cite{joseph2020koopmanvon-960} is to define a grid over the $N$-dimensional hypercube with $m$ points per axis.
The dimension of the system is then $m^N$, which does not allow an amplitude encoding on $\poly\log(N)$ qubits.
We choose an alternative truncation approach as first proposed in~\cite{engel2021linear-e8c}, and consider the dynamics projected onto a space with a maximum number of bosons $m$
\begin{equation}
\label{eq:global-hermite-space}
\mathcal H_{\leq m}
=
\operatorname{span}\{h_\alpha:\ |\alpha|\leq m\}.
\end{equation}
The projected operator is
\begin{equation}
\label{eq:galerkin-projected-operator}
L_m=\Pi_m L\Pi_m.
\quad
\Pi_m:L^2(\mathbb R^N)\to \mathcal H_{\leq m}.
\end{equation}
The Galerkin dynamics are then the finite-dimensional system
\begin{equation}
\label{eq:galerkin-dynamics}
\partial_t u_m(t)=L_m u_m(t),
\qquad
u_m(0)=\Pi_m u_0.
\end{equation}
Thus, the dimension of $\mathcal H_{\leq m}$ is the number of ways to distribute at most $m$ excitations among $N$ modes.
By the stars-and-bars combinatorial formula,
\begin{equation}
\label{eq:global-hermite-dimension}
D_{N,m}
=
\dim \mathcal H_{\leq m}
=
\binom{N+m}{m} \sim \frac{N^m}{m!}.
\end{equation}
Amplitude encoding a state of that dimensionality is possible with $\Theta(m\log(N))$ qubits, and allows for space-efficiency for $N=2^n$ and $m=\operatorname{poly}(n)$. 
We call this the global truncation strategy.

Up to this point we have been using existing concepts to generalise the bosonic encoding to any linear high-dimensional PDE. 
We now investigate whether the global cut-off preserves accuracy, for a restricted class of linear PDEs with strong regularity and stability assumptions.

\subsection{Convergence}
\label{subsec:convergence}
The projection of the infinite-dimensional system onto a finite one is an approximation.
In order to control the accuracy, we must find conditions under which the solution remains concentrated in subspaces with low bosonic numbers. 
This can be quantified using a number operator centric regularity. 
For $s\geq 0$, we define the oscillator $s$-regularity as
\begin{equation}
\label{eq:oscillator-sobolev-norm}
\|u\|_{s,\mathrm{osc}}
:=
\|(I+\hat n)^{s/2}u\|_{L^2(\mathbb R^N)}.
\end{equation}
If we know a bound on this norm, and if the truncated dynamics are stable, we can control the fidelity of the truncated dynamics to the infinite-dimensional dynamics, as we prove in \Cref{app:finite-time-galerkin-proof}. 
\begin{theorem}[Dynamics convergence]
\label{thm:finite-time-galerkin-truncation}
Let $m$ be the truncation order, and $L=P(\hat q,\hat p)$ be a polynomial of total degree $d$ in the position $\hat q$ and momentum $\hat p$ operators of $N$ modes. 
Then there exists a constant $C_P\in\mathbb{R}^+$ such that
$$
\|Lv\|\le C_P\|v\|_{d,\mathrm{osc}}\,.
$$
Let $u(t)$ solve $\partial_t u=Lu$ on $[0,T]$, and let $u_m(t)$ solve \Cref{eq:galerkin-dynamics}. 
Assume that there exists $s>0$ such that the exact solution has bounded oscillator $(s+d)$-regularity as defined in \Cref{eq:oscillator-sobolev-norm} on $[0,T]$, that is there exists a constant $C_{s+d}\in\mathbb{R}^+$ such that
$$
\sup_{0\leq \tau\leq T}
\|u(\tau)\|_{s+d,\mathrm{osc}} \leq C_{s+d} < +\infty\,.
$$
In addition assume that the projected dynamics are uniformly stable on $[0,T]$, that is there exists a constant $C_{T}\in\mathbb{R}^+$ such that
$$
\|e^{tL_m}\|\le C_T,\qquad 0\le t\le T.
$$
Then there exists a constant independent of $m$ such that for all $0\leq t\leq T$ the approximation accuracy is polynomial in $m$
\begin{equation}
\label{eq:finite-time-galerkin-simple}
\|u(t)-u_m(t)\|
\leq
(m+1)^{-s/2}
\left(
1
+
C_T C_P t
\right)C_{s+d}.
\end{equation}
\end{theorem}

As we argue in \Cref{app:finite-time-galerkin-proof}, the stability assumption is not automatic for $L=P(\hat q,\hat p)$, as some polynomial generators can exhibit finite-time singularities in their associated Heisenberg dynamics.
However, stability is automatic in common unitary or contractive cases, for example $L=-iH$ with $H$ self-adjoint or $L=-H$ with $H \succeq 0$. 

We also prove a convergence result for ground state problems.
Ground-state problems in that context encapsulate steady state solutions of linear PDEs, as shown in \Cref{sec:fokker-planck}.
We consider the special case where the operator $L$ is Hermitian and positive, with its ground energy equal to 0, for which we give examples in \Cref{sec:fokker-planck}.
We prove the following theorem in \Cref{app:ground-state-consistency}.

\begin{theorem}[Statics convergence]
\label{thm:statics-convergence}
Let $H$ be a nonnegative self-adjoint operator on $L^2(\mathbb R^N)$ with normalized ground state
$\phi$ satisfying
$
H\phi=0.
$
Assume that $H$ has spectral gap $\gamma>0$, namely
$$
\sigma(H)\subset \{0\}\cup[\gamma,\infty).
$$
Let
$$
H_m:=\Pi_m H\Pi_m
$$
on $\mathcal H_{\le m}$, and let $\lambda_{0,m}$ and $\lambda_{1,m}$ be its two lowest eigenvalues and $\phi_m$ a normalized ground state.
Let
$
\gamma_m:=\lambda_{1,m}-\lambda_{0,m}.
$

Assume that, for some $d\ge 0$, there exists a constant $C_H\in \mathbb{R}^+$ such that
\begin{equation}
\|Hv\|\le C_H\|v\|_{d,\mathrm{osc}}\,.
\label{eq:H-osc-control}
\end{equation}
Assume also that for some $s>0$, there exists a constant $C_{s+d}\in \mathbb{R}^+$ such that
$$
\|\phi\|_{s+d,\mathrm{osc}} \leq C_{s+d} <+\infty\,.
$$
Then, for all sufficiently large $m$, defining
\begin{equation}
E_m
:=
2C_HC_{s+d}(m+1)^{-s/2},
\label{eq:Em-definition}
\end{equation}
one has
\begin{equation}
0\le \lambda_{0,m}\le E_m .
\label{eq:lambda0m-bound}
\end{equation}
\begin{equation}
\|\phi_m-\phi\|
\le
\left(\frac{2E_m}{\gamma}\right)^{1/2}.
\label{eq:ground-state-projection-error}
\end{equation}
\begin{equation}
\gamma_m
\ge
\gamma-3E_m .
\label{eq:truncated-gap-bound}
\end{equation}
In particular, if $E_m\le \gamma/6$, then $\gamma_m\ge \gamma/2$.
\end{theorem}

Having established convergence of the global truncation, we now ask whether the truncated operator $L_m$ can be block encoded efficiently on qubits. 
This is the subject of the next section.

\section{Qubit encoding of the global Hermite truncation}
In this section, we introduce a qubit encoding of the globally truncated space $\mathcal H_{\leq m}$. 
The encoding uses $m$ registers and represents bosonic states by states that are invariant under permutations of these registers. 
We first define the encoding, then construct block encodings of the annihilation, creation, position, momentum, and number operators. Finally, we combine these building blocks to block encode structured polynomial differential operators.
\label{sec:qubit-encoding}

\subsection{Bosonic indistinguishability and the symmetric subspace}
\label{subsec:bosonic-indistinguishability}
We encode states containing at most $m$ bosons using $m$ registers, one for each possible boson. 
Each register stores one of $N$ labels, indicating the mode occupied by that boson, or an empty label $\emptyset$ when fewer than $m$ bosons are present. 
The state space of one register is therefore
\begin{equation}
\mathcal K_N
=
\operatorname{span}
\{
\ket{\emptyset},\ket{1},\ldots,\ket{N}
\}
\simeq
\mathbb C^{N+1},
\end{equation}
and the full register space is $\mathcal K_N^{\otimes m}$.

Bosons are indistinguishable, so exchanging two registers must not change the encoded state. 
We therefore restrict to states that are invariant under any permutation of the $m$ registers. 
These states form the symmetric subspace
\begin{equation}
\label{eq:symmetric-register-space}
\operatorname{Sym}^m(\mathcal K_N)
\subset
\mathcal K_N^{\otimes m},
\end{equation}
where $\operatorname{Sym}^m(\mathcal K_N)$ denotes the set of states unchanged by permutations of the registers.

We write $\ket{\alpha}_F$ for a state in the occupation-number representation, where $\alpha=(\alpha_1,\ldots,\alpha_N)$ and $\alpha_j$ is the number of bosons in mode $j$.
We write $\ket{b}_Q=\ket{b_1,\ldots,b_m}_Q$ for a state in the register representation, where each $b_r$ is either a mode label or the empty label $\emptyset$.

For example, consider at most three bosons distributed among four modes, so that $N=4$ and $m=3$.
The state with two bosons in mode $1$ and one boson in mode $4$ is encoded as
$$
\ket{2,0,0,1}_F
\mapsto
\frac{
\ket{1,1,4}_Q+
\ket{1,4,1}_Q+
\ket{4,1,1}_Q
}{\sqrt 3}.
$$
There is no empty slot because the total occupation is already $m=3$. 
By contrast, the state with one boson in mode $1$ and one boson in mode $4$ contains only two bosons. It is encoded as
\begin{align*}
    &\ket{1,0,0,1}_F
    \mapsto
    \frac{1}{\sqrt 6} \Big(
    \ket{1,4,\emptyset}_Q+
    \ket{1,\emptyset,4}_Q+\\
    &\ket{4,1,\emptyset}_Q+
    \ket{4,\emptyset,1}_Q+
    \ket{\emptyset,1,4}_Q+
    \ket{\emptyset,4,1}_Q\Big)
\end{align*}
More generally, for $\alpha$ satisfying $|\alpha|\leq m$, we form a list containing $\alpha_j$ copies of each label $j$ and $m-|\alpha|$ copies of $\emptyset$. 
The encoded state is the normalized uniform superposition over all distinct permutations of this list. 
This defines the map $V_m$. Its explicit expression and proof are given in \Cref{app:symmetric-encoding-isometry}.
\begin{equation}
\label{eq:symmetric-encoding-isometry}
V_m:\mathcal H_{\leq m}\longrightarrow \operatorname{Sym}^m(\mathcal K_N),
\end{equation}
\begin{proposition}[Symmetric encoding]
\label{prop:symmetric-encoding-global-cut-off}
The map $V_m$ in \Cref{eq:symmetric-encoding-isometry} is an isometry from the globally truncated Hermite space $\mathcal H_{\leq m}$ onto $\operatorname{Sym}^m(\mathcal K_N)$. 
\end{proposition}

From now on, when the cut-off $m$ is fixed, we write $|\alpha\rangle$ for the encoded state $V_m|\alpha\rangle_F$.
The next step is to represent the operators acting on this space, while preserving the permutation symmetry on the $m$ registers.

\subsection{Encoding of primitive bosonic operators}
\label{subsec:encoded-bosonic-primitives}

In this subsection, we construct qubit operators representing the truncated annihilation, creation, position, momentum, and number operators introduced in \Cref{subsec:background}. 
Since encoded states are invariant under permutations of the registers, these operators must act identically on all $m$ registers. The detailed constructions are given in \Cref{app:encoded-bosonic-primitives} and yield the following proposition.

\begin{proposition}[Encoded bosonic primitives]
\label{prop:encoded-bosonic-primitives}
Under the symmetric encoding $V_m$, the globally truncated annihilation, creation, position, momentum, and number operators are represented by operators $\hat A_j$, $\hat A_j^\dagger$, $\hat Q_j$, $\hat P_j$, and $\hat N$. These operators admit block encodings with gate complexity polynomial in $m$ and $\log N$. Their normalizations scale as $O(\sqrt m)$ for $\hat A_j$, $\hat A_j^\dagger$, $\hat Q_j$, and $\hat P_j$, and as $O(m)$ for $\hat N$.
\end{proposition}

The construction can be understood most easily for the annihilation operator $\hat a_j$, which removes one excitation from mode $j$. 
In the register encoding, this corresponds to replacing one label $j$ by the empty label $\emptyset$. Because the label $j$ may occur in any register, we sum this transition over all registers:
\begin{equation}
B_j
:=
\sum_{r=1}^m
\left(\ketbra{\emptyset}{j}\right)_r.
\end{equation}
The action of $B_j$ contains an additional factor depending on the number of empty registers. We count these registers using
\begin{equation}
\hat N_{\emptyset}
:=
\sum_{r=1}^m
\left(\ketbra{\emptyset}{\emptyset}\right)_r
\end{equation}
and remove the additional factor by a diagonal normalization, as detailed in \Cref{app:encoded-bosonic-primitives}. This gives the encoded annihilation operator $\hat A_j$; the creation operator $\hat A_j^\dagger$ is obtained similarly. The remaining operators follow from
\begin{equation}
\hat Q_j
=
\frac{\hat A_j+\hat A_j^\dagger}{\sqrt 2},
\quad
\hat P_j
=
\frac{\hat A_j-\hat A_j^\dagger}{i\sqrt 2},
\quad
\hat N
=
mI-\hat N_{\emptyset}.
\end{equation}

The encoding of bosonic operators over exponentially many modes is inspired by \cite{barthe2025gate-based-8f8}. 
Whereas that work is tailored to Gaussian bosonic dynamics, the present construction uses $m$ permutation-symmetric registers to represent states with up to $m$ excitations, thereby supporting non-Gaussian states and operators.
Bravyi et al.~\cite{bravyi2025quantum-101} instead store the occupied mode labels in a sorted, padded list. 
Both approaches exploit low bosonic occupation, although Bravyi et al. use a weighted cutoff. 
Their sorted representation is tailored to the cubic operators arising in their dynamics. 
In the present encoding, each elementary bosonic operation is expressed as a sum of single-register transitions, which allows more general polynomial operators to be constructed modularly.

\subsection{Block encoding of structured polynomials in position and momentum operators}
\label{subsec:block-encoding-polynomial-operators}
The previous subsection gives block encodings for the elementary bosonic operators. 
We now compose these primitives to block encode the truncated differential operator. 
An arbitrary polynomial may contain exponentially many coefficients in $N$. 
We therefore restrict to operators given as a polynomial-size sum of products of structured linear combinations $l$ of position and momentum operators.

If $v,w\in \mathbb C^N$, define
\begin{equation}
\label{eq:encoded-linear-form}
l(v,w)
=
\sum_{j=1}^N v_j \hat Q_j+\sum_{j=1}^N w_j \hat P_j.
\end{equation}
When the normalized states proportional to $v$ and $w$ are efficiently preparable, $l(v,w)$ admits an efficient block encoding with normalization $O(\sqrt m(\|v\|+\|w\|))$. 
This is the primitive used below for polynomial differential operators.

\begin{definition}[Structured operator polynomial]
\label{def:structured-operator-polynomial}
Let $N=2^n$. 
A structured operator polynomial over $N$ modes is a polynomial of constant degree $d\in O(1)$ of the form
\begin{equation}
\label{eq:structured-operator-polynomial-definition}
P(\hat q,\hat p)
=
\sum_{r=1}^R
c_r
\prod_{\ell=1}^{k_r}
\left(
a_{r,\ell}I
+
v_{r,\ell}\cdot \hat q
+
w_{r,\ell}\cdot \hat p
\right),
\end{equation}
where $k_r\leq d$, $c_r,a_{r,\ell}\in\mathbb C$ and $v_{r,\ell},w_{r,\ell}\in\mathbb C^N$. 
We assume in addition that:
\begin{enumerate}
    \item The integer $R$; the magnitude of the scalars $c_r$ and $a_{r,\ell}$; the vector norms $\|v_{r,\ell}\|$ and $\|w_{r,\ell}\|$ are all bounded by $\operatorname{poly}(n)$;
    \item for every $(r,\ell)$, the normalized amplitude-encoded states of the non-zero vectors $v_{r,\ell}$ and $w_{r,\ell}$
    can be prepared by circuits of size $\operatorname{poly}(n)$. 
\end{enumerate}
\end{definition}
We show in \Cref{app:block-encoding-polynomial-operators} that when the operator of interest is a structured polynomial, the operator admits an efficient block encoding.

\begin{theorem}[Structured polynomial block encoding]
\label{thm:structured-polynomial-block-encoding}
Let $N=2^n$ and $m\geq 0$, and let
$$
L_m=\Pi_m P(\hat q,\hat p)\Pi_m
$$
where $P$ is a structured polynomial of degree $d$ over $N$ variables as in \Cref{def:structured-operator-polynomial}.
The encoded operator $V_m L_m V_m^\dagger$ admits a block encoding with normalization $\alpha\in O(C_P(m+d)^{d/2})$ with $C_P \in O(\poly(n))$ and gate complexity $ O(\operatorname{poly}(m,n))$.
\end{theorem}

At this stage we have shown that any structured constant-degree polynomial of position and momentum operators over exponentially many modes admits an efficient block encoding. 
Equivalently, a structured constant-degree polynomial in multiplication by $x_j$ and differentiation with respect to $x_j$ over exponentially many variables admits an efficient block encoding for its Hermite representation. 
In the next subsection, we show how relevant information may be extracted from such encoded states.

\subsection{Readout: observables and marginals}
\label{subsec:readout-observables-marginals}
In this subsection, we ask what useful quantities can be extracted efficiently given an encoded state
$$
\ket{\psi_m}
=
\sum_{|\alpha|\leq m} c_\alpha \ket{\alpha}\,.
$$
When $N$ is exponentially large, even writing down one value for every coordinate requires exponential resources.
We can nevertheless estimate low-order moments of a fixed number of coordinates and, when $\ket{\psi_m}$ represents a probability amplitude, reconstruct constant-dimensional marginal distributions.
We show that the following outputs can be obtained efficiently, and we provide details in \Cref{app:readout-observables-marginals}.
\begin{proposition}[Accessible readout]
\label{prop:accessible-readout}
Assume that $\ket{\psi_m}$ can be prepared efficiently. Then the following quantities are efficiently accessible.
\begin{enumerate}
    \item expectation values $\langle Q(\hat q,\hat p)\rangle$ of structured polynomial observables $Q$ as in \Cref{def:structured-operator-polynomial};
    \item expectations of the number operator $\langle \hat N\rangle$;
    \item marginal distributions of any fixed number of coordinates when $|\psi_m(x)|^2$ is a probability density.
\end{enumerate}
\end{proposition}

The expectation of the number operator is especially useful as a truncation diagnostic. 
The expectation $\langle\hat N\rangle$ measures the average total Hermite degree. 
A value close to $m$ indicates that the cut-off may be too small.

\section{Fokker--Planck equation and Bayesian inverse problems}
\label{sec:fokker-planck}
In this section, we apply the Hermite truncation and qubit encoding developed above to Bayesian inverse problems. 
We first recall the standard transformation that maps a reversible Fokker--Planck operator to a non-negative Hermitian operator whose ground state is the square root of its stationary density. 
We then show that, after whitening a Gaussian prior, a Bayesian posterior can be expressed in this form. 
We then identify conditions under which the resulting Hamiltonian is a structured polynomial operator that can be block encoded efficiently, and analyse the cost of preparing a quantum state whose position-space amplitude is the square root of the posterior density.

\subsection{Fokker--Planck equations and Gibbs-preserving drifts}
\label{subsec:general-fokker-planck}

The material in this subsection is standard, found for example in~\cite{risken1996fokker-planck-aa7}, we restate it using conventions of this work. 
We consider Fokker--Planck equations of the form
\begin{equation}
\label{eq:general-fokker-planck}
\partial_t \rho
=
\nabla\cdot(\rho F)+\Delta \rho,
\end{equation}
where $F:\mathbb R^N\to\mathbb R^N$. With this convention, the drift of the corresponding stochastic differential equation is $-F$.
Following the decomposition described in~\cite[Sec.~6.3]{risken1996fokker-planck-aa7}, suppose that
\begin{equation}
\label{eq:gibbs-preserving-decomposition}
F=\nabla V+G,
\qquad
\nabla\cdot(e^{-V}G)=0,
\end{equation}
where the potential $V$ satisfies
\begin{equation}
Z:=\int_{\mathbb R^N}e^{-V(x)}\,dx<+\infty.
\end{equation}
Then
\begin{equation}
\label{eq:gibbs-stationary-density}
\rho_\infty(x):=Z^{-1}e^{-V(x)}
\end{equation}
is a stationary density of \Cref{eq:general-fokker-planck}. Indeed,
\begin{equation}
\nabla\cdot(\rho_\infty F)+\Delta\rho_\infty
=
Z^{-1}\nabla\cdot(e^{-V}G)
=
0.
\end{equation}

The decomposition in \Cref{eq:gibbs-preserving-decomposition} separates the gradient drift $\nabla V$ from a transport field $G$ that preserves the same stationary density. 
The field $G$ can modify the transient evolution and the rate of convergence to equilibrium, but it does not change $\rho_\infty$. 
In the next subsection, we recall the classical transformation under which the gradient contribution becomes imaginary-time evolution generated by a non-negative Hermitian operator, while the contribution from $G$ becomes anti-Hermitian.

After translating sign conventions, the setting of Bravyi et al.~\cite{bravyi2025quantum-101} corresponds to a diagonal quadratic potential of the following form, with $q>0$ and $\lambda_i>0$,
\begin{equation}
V_{\text{Bravyi}}(x)=q^{-1}\sum_i\lambda_i x_i^2\,.
\end{equation}
They consider a quadratic transport field $G$ satisfying a divergence-free condition corresponding to the Gibbs-preserving condition used here.
Their invariant density is therefore the explicitly known Gaussian
\begin{equation}
\rho_\infty(x)
\propto
\exp\left(-q^{-1}\sum_i\lambda_i x_i^2\right),
\end{equation}
and their analysis focuses on the transient dynamics. 
They apply the corresponding transformation to the backward Kolmogorov equation for observables, whereas we use the forward Fokker--Planck equation for densities. 
Their truncation similarly retains a low-occupation sector of the bosonic space. 
Our framework allows the stationary potential to extend beyond diagonal quadratic forms, subject to the structured-polynomial assumptions introduced above, and uses a symmetric-register encoding that supports modular block encodings of general fixed-degree polynomial operators.

\subsection{Ground-state transform}
We now recall the standard change of variables that maps the stationary density to the ground state of a Hermitian operator~\cite[Sec.~6.3]{risken1996fokker-planck-aa7}. 
Let
\begin{equation}
\phi(x):=\sqrt{\rho_\infty(x)}=Z^{-1/2}e^{-V(x)/2},
\end{equation}
and define
\begin{equation}
\psi(t,x):=\frac{\rho(t,x)}{\phi(x)}\,.
\end{equation}
Under this change of variables, the Fokker--Planck equation becomes
\begin{equation}
\partial_t\psi=-H_V\psi+iH_G\psi,
\label{eq:ground-state-transformed-equation}
\end{equation}
where $H_V$ is the Hermitian operator associated with the gradient drift $\nabla V$, while $iH_G$ is the anti-Hermitian operator associated with the transport field $G$:
\begin{align}
H_V
&=
-\Delta+\frac14|\nabla V|^2-\frac12\Delta V,
\\
iH_G
&=
G\cdot\nabla+\frac12(\nabla\cdot G).
\end{align}
The derivation is given in \Cref{app:ground-state-transform-fokker-planck}.
This change of variables differs from that of \cite{wu2026from-49e}, where the central object is a mixed state over Brownian realisations of $\sqrt{\rho(x,t)}$.

Moreover, $H_V$ is non-negative because it factorizes as
\begin{equation}
\label{eq:fp-hamiltonian-factorization}
H_V
=
\sum_{j=1}^N
\left(
-\partial_{x_j}
+
\frac12\partial_{x_j}V
\right)
\left(
\partial_{x_j}
+
\frac12\partial_{x_j}V
\right)\,.
\end{equation}
Together with $H_V\phi=0$, this shows that $\phi$ is a ground state of $H_V$.
\label{subsec:ground-state-transform}

\subsection{Potentials with a $\|x\|^2$ term}
When $V(x)=\|x\|^2$, the potential combines with the diffusion term to yield the number operator:
\begin{equation}
\label{eq:quadratic-potential-number-operator}
H_V
=
-\Delta+\|x\|^2-N
=
2\hat n.
\end{equation}
Therefore, the Hermite basis diagonalizes $H_V$.

More generally, suppose that
\begin{equation}
\label{eq:potential-quadratic-plus-perturbation}
V(x)=\|x\|^2+F_0(x),
\qquad
F_0(x)\geq 0.
\end{equation}
The non-negativity of $F_0$ ensures that $e^{-V(x)}$ decays at least as fast as a Gaussian, and hence that the normalization constant $Z$ is finite. In this case,
\begin{equation}
\label{eq:H-number-plus-potential}
H_V
=
2\hat n+W_{F_0}(\hat q),
\end{equation}
where
\begin{equation}
\label{eq:WF0-definition}
W_{F_0}(x)
=
x\cdot\nabla F_0(x)
+
\frac14|\nabla F_0(x)|^2
-
\frac12\Delta F_0(x).
\end{equation}
The appearance of the number operator in $H_V$ motivates the use of the Hermite basis and gives intuition for why the global Hermite truncation can converge. In the dynamics picture, the term $2\hat n$ penalizes components with larger boson number, causing them to decay more rapidly under imaginary-time evolution. In the static picture, it assigns a larger energy to states with high boson number.

\subsection{Bayesian inverse problems as ground-state problems}
\label{subsec:bayesian-inverse-problems-ground-state}
We now connect the Fokker--Planck equation to Bayesian inverse problems~\cite{stuart2010inverse-b3f}. 
Let $u\in\mathbb R^N$ be an unknown parameter and $y\in\mathbb R^M$ the observed variables.
We are given a forward map $g:\mathbb R^N\to\mathbb R^M$ that predicts the observations from the unknown parameter.
We assume $u$ has a Gaussian prior and the observations are affected by Gaussian noise.
\begin{equation}
\label{eq:bayesian-observation-model}
y=g(u)+\eta
\,,
\eta\sim \mathcal N(0,\Gamma)
\,,
u\sim \mathcal N(\mu,\Sigma).
\end{equation}
We introduce rescaled whitened coordinates $x$ so that the prior density is proportional to $e^{-\|x\|^2}$:
\begin{equation}
\label{eq:whitened-coordinate}
u=\mu+\sqrt 2\,\Sigma^{1/2}x\,.
\end{equation}
Bayes' formula gives the posterior density
\begin{equation}
\label{eq:posterior-density-whitened}
\pi(x|y)
=
Z_y^{-1}e^{-V_y(x)},
\quad
V_y(x)
=
\|x\|^2+F_y(x),
\end{equation}
where
\begin{equation}
\label{eq:inverse-problem-potential}
F_y(x)
=
\frac12
\left\|\Gamma^{-1/2}(
y-g(\mu+\sqrt 2\,\Sigma^{1/2}x))
\right\|^2.
\end{equation}
Thus the Gaussian prior yields the $\|x\|^2_2$ term of the potential, and the observation fit yields the perturbation.

From here the derivation follows as in the previous subsection with $F_y(x)$ in the place of $F_0(x)$.
Defining
\begin{equation}
\phi_y(x):=\sqrt{\pi(x|y)},
\end{equation}
for any observable $f$, posterior expectations can be written as
\begin{equation}
\label{eq:posterior-expectation-ground-state}
\mathbb E[f(x)]
=
\int f(x)\pi(x|y)\,dx
=
\langle \phi_y,f(\hat q)\phi_y\rangle .
\end{equation}
Thus posterior computation can be reduced to preparing or approximating $\phi_y$ and estimating observables in this state. 
The global Hermite truncation gives the finite-dimensional approximation, while the block-encoding results of \Cref{subsec:block-encoding-polynomial-operators} apply whenever $W_y$ has a structured polynomial representation.

\subsection{Asymptotic complexity}
\label{subsec:asymptotics-bayes}
We now combine the truncation result of \Cref{subsec:convergence}, the block encoding of \Cref{subsec:block-encoding-polynomial-operators}, and a standard ground-state-preparation algorithm to estimate the overall complexity.

Let
\begin{equation}
H_y:=2\hat n+W_y(\hat q)\,.
\end{equation}
Assume that $W_y(\hat q)$ has a structured operator-polynomial representation of degree $d_W\in O(1)$ in the sense of \Cref{def:structured-operator-polynomial}. A sufficient condition on the forward map $g$ is given in \Cref{app:structured-forward-map-posterior-hamiltonian}.

Since $H_y$ is a polynomial of degree
\begin{equation}
d:=\max\{d_W,2\},
\end{equation}
there exists a constant $C_H$ such that
\begin{equation}
\|H_yv\|
\leq
C_H\|v\|_{d,\mathrm{osc}}.
\end{equation}
For polynomial $F_y$, the ground state has finite oscillator regularity of every finite order. For some $s>0$, define
\begin{equation}
\|\phi_y\|_{s+d,\mathrm{osc}}
:=
C_{s+d}
<
+\infty.
\end{equation}
Although this quantity is finite its dependence on the problem size is instance-dependent. 
An efficient algorithm additionally requires $C_H$ and $C_{s+d}$ to grow at most polynomially in $\log N$.

Let $\epsilon>0$ be the target state-preparation error. We also assume that a lower bound $\gamma>0$ on the spectral gap of $H_y$ is known, and that an efficiently preparable initial state $\psi_{\mathrm{init}}$ satisfies
\begin{equation}
|\langle\psi_{\mathrm{init}},\phi_y\rangle|
\geq
\xi.
\end{equation}
We must choose a truncation order $m$ such that 
$$
\|\phi_y-\phi_{y,m}\|\le \epsilon,
$$
$$
\gamma_m\ge \gamma/2,
$$
$$
|\langle \psi_{\rm init},\phi_{y,m}\rangle|\ge \xi/2.
$$
Using \Cref{thm:statics-convergence}, this is ensured by requiring
$$
E_m \le \frac{\gamma}{8}\min\{1,\epsilon^2,\xi^2\},
$$
where
$$
E_m = 2C_HC_{s+d}(m+1)^{-s/2}.
$$
We choose the truncation level accordingly
\begin{equation}
m_\epsilon
:=
\left\lceil
\left(
\frac{16C_HC_{s+d}}
{\gamma\min\{1,\epsilon^2,\xi^2\}}
\right)^{2/s}
\right\rceil .
\label{eq:choice-of-m-ground-state}
\end{equation}
Using \Cref{thm:structured-polynomial-block-encoding} for our choice of $m_{\epsilon}$ we have a block encoding of gate complexity $G(m_\epsilon) \in \poly(m_\epsilon)$ and normalization
\begin{equation}
\alpha_{m_\epsilon}
\in
\Theta\!\left(C_P(m_\epsilon+d)^{d/2}\right).
\label{eq:posterior-hamiltonian-normalization}
\end{equation}

Then using optimal scaling of ground state preparation from~\cite{lin2020near-optimal-4ea}, $\phi_y$ can be prepared to accuracy $\epsilon$ with gate complexity
\begin{equation}
\widetilde O
\left(
G(m_\epsilon)C_P\,
\frac{(m_\epsilon+d)^{d/2}}{\gamma\,\xi}
\log\frac{1}{\epsilon}
\right).
\label{eq:posterior-ground-state-preparation-complexity}
\end{equation}
The space complexity, i.e. the required number of qubits is
\begin{equation}
\tilde\Theta(m_\epsilon\log N).
\label{eq:posterior-ground-state-space-complexity}
\end{equation}
In addition, information about this ground state can be extracted efficiently according to \Cref{prop:accessible-readout}.
Although the algorithm scales poly-logarithmically with $N$ when $C_{s+d}$ and $C_H$ are polynomial in $\log N$, its overall complexity follows a power law in $1/\epsilon$.

\subsection{Spectral gap and hardness}
\label{subsec:spectral-gap-hardness}

\Cref{thm:structured-polynomial-block-encoding} gives conditions under which the posterior Hamiltonian can be block encoded efficiently. 
This is not sufficient for an efficient inverse-problem algorithm: ground-state preparation also depends on the spectral gap and on the overlap with an efficiently preparable initial state~\cite{lin2020near-optimal-4ea}.

When the forward map $g$ is linear, the spectral gap is bounded below, but the inverse problem is also classically tractable, as shown in \Cref{app:linear-gaussian-classical-reduction}.
When $g$ is quadratic, the data-misfit term $F_y$ is quartic and the posterior may have several local minima.
For the class of quadratic-square potentials considered in \Cref{app:quadratic-square-spectral-gap}, we prove that the corresponding Hamiltonian has a strictly positive spectral gap for every fixed instance.
However, this is only an existence result and does not provide a lower bound that remains polynomially large as the problem size increases.
When the minima are separated by large barriers, the gap may be exponentially small in the barrier height~\cite{helffer2005hypoelliptic-658,bovier2005metastability-771}.

This limitation is unavoidable in general.
In \Cref{app:ising-hardness-quartic-potentials}, we show that quartic potentials can concentrate the posterior near the vertices of a hypercube and reproduce the probability distribution of a discrete spin system.
In that regime, the posterior normalization constant approximates an Ising partition function, while some posterior expectations approximate Ising Gibbs observables. 
Exact evaluation of Ising partition functions is $\#$P-hard in broad parameter regimes, and approximation with external fields can also be hard~\cite{jaeger1990computational-6cf,goldberg2006complexity-d2d}.
One should therefore not expect a polynomial-time algorithm for arbitrary quartic inverse problems.
Efficient block encoding is a property of the Hamiltonian representation; efficient state preparation additionally requires a sufficiently large spectral gap and a suitable preparation strategy.

We now give an additional assumption under which both difficulties can be addressed.
Suppose that there exists $\beta<2$ such that, for all $x\in\mathbb R^N$,
\begin{equation}
\nabla^2F_y(x)\succeq-\beta I.
\end{equation}
Define the interpolating potential
\begin{equation}
V_\tau(x):=\|x\|^2+\tau F_y(x),
\qquad
0\leq\tau\leq1,
\end{equation}
and let $H_\tau$ be the corresponding Hamiltonian.
As proved in \Cref{app:hessian-lower-bound-gap},
\begin{equation}
\gamma(H_\tau)\geq2-\beta\tau\geq2-\beta=:\kappa>0.
\end{equation}
The spectral gap therefore remains open along the entire path.
At $\tau=0$, the ground state is the bosonic vacuum $\ket{0}$.
The posterior state can consequently be prepared by evolving adiabatically from $\tau=0$ to $\tau=1$~\cite{albash2018adiabatic-e7f}, with polynomial dependence on $1/\kappa$ and $1/\epsilon$, in addition to the cost of implementing the Hamiltonian path.

\acknowledgments
AB thanks Matteo Lostaglio for crucial feedback and the adiabatic state preparation idea. 
AB thanks Paul Mannix for extensive feedback. 

\printbibliography

\clearpage
\onecolumn
\appendix
\section{Truncation estimates for the global Hermite cut-off}
\label{app:global-hermite-truncation}

\subsection{Ordering convention for polynomial differential operators}
\label{app:ordering-polynomial-differential-operators}

Note that the operators $\hat q_j$ and $\hat p_j$ do not commute.
Therefore, we need to fix an ordering convention. Let
$$
L
=
\sum_{|\beta|\le r}
a_\beta(x)\partial^\beta
$$
be a finite-order differential operator with polynomial coefficients. Since multiplication by $x_j$ is $\hat q_j$ and $\partial_{x_j}=i\hat p_j$, one may choose the ordering in which all position operators are placed to the left of all momentum operators:
\begin{equation}
\label{eq:ordered-polynomial-qp}
P(\hat q,\hat p)
=
\sum_{|\beta|\le r}
a_\beta(\hat q)(i\hat p)^\beta .
\end{equation}
Then $L=P(\hat q,\hat p)$ on $\mathcal S(\mathbb R^N)$. 

\subsection{Hermite tails}
\label{app:hermite-tail-proof}
We show that functions with a bounded oscillator Sobolev norm as in \Cref{eq:oscillator-sobolev-norm} have controlled Hermite expansion tails.
\begin{proposition}[Hermite tail bound]
\label{prop:hermite-tail-bound}
Let $s\geq 0$. Then
\begin{equation}
\label{eq:hermite-tail-bound}
\|(I-\Pi_m)u\|
\leq
(m+1)^{-s/2}\|u\|_{s,\mathrm{osc}}.
\end{equation}
\end{proposition}
Note that this is only a non-trivial when $\|u\|_{s,\mathrm{osc}} < +\infty$, that is  $u$ is in the domain of $(I+\hat n)^{s/2}$.
\begin{proof}
    Write
    $$
    u=\sum_{\alpha\in\mathbb N^N}u_\alpha h_\alpha.
    $$
    Since $\hat n h_\alpha=|\alpha|h_\alpha$, we have
    $$
    \|u\|_{s,\mathrm{osc}}^2
    =
    \sum_{\alpha\in\mathbb N^N}
    (1+|\alpha|)^s |u_\alpha|^2.
    $$
    Moreover,
    $$
    \|(I-\Pi_m)u\|^2
    =
    \sum_{|\alpha|>m}|u_\alpha|^2.
    $$
    For $|\alpha|>m$, one has $(1+|\alpha|)^{-s}\leq (m+1)^{-s}$. 
    Therefore
    $$
    \sum_{|\alpha|>m}|u_\alpha|^2
    =
    \sum_{|\alpha|>m}
    (1+|\alpha|)^{-s}(1+|\alpha|)^s|u_\alpha|^2
    \leq
    (m+1)^{-s}
    \|u\|_{s,\mathrm{osc}}^2.
    $$
    Taking the square root gives \Cref{eq:hermite-tail-bound}.
\end{proof}

\subsection{Polynomial operators and oscillator Sobolev norms}
\label{app:polynomial-operator-control}

We use the following standard estimate.

\begin{lemma}[Polynomial control by the number operator]
\label{lem:polynomial-control-number}
Let $P(\hat q,\hat p)$ be an ordered polynomial of total degree $d$. Then, for every $s\geq 0$, there exists a constant $C_{P,s}$ such that
\begin{equation}
\label{eq:polynomial-control-number}
\|P(\hat q,\hat p)v\|_{s,\mathrm{osc}}
\leq
C_{P,s}\|v\|_{s+d,\mathrm{osc}}
\end{equation}
for every $v$ in the Hermite span.
\end{lemma}

\begin{proof}
It is enough to prove the estimate for monomials in $\hat q_j$ and $\hat p_j$. Each $\hat q_j$ or $\hat p_j$ is a linear combination of $\hat a_j$ and $\hat a_j^\dagger$. Thus it suffices to control products of creation and annihilation operators.
For a Hermite expansion $v=\sum_\alpha v_\alpha h_\alpha$,
$$
\hat a_j v
=
\sum_{\alpha:\alpha_j\geq 1}
\sqrt{\alpha_j}\,v_\alpha h_{\alpha-e_j}.
$$
Since $\alpha_j\leq |\alpha|$, this gives
$$
\|\hat a_j v\|
\leq
\|\hat n^{1/2}v\|
\leq
\|v\|_{1,\mathrm{osc}}.
$$
Similarly,
$$
\hat a_j^\dagger v
=
\sum_\alpha
\sqrt{\alpha_j+1}\,v_\alpha h_{\alpha+e_j},
$$
and $\alpha_j+1\leq |\alpha|+1$, so
$$
\|\hat a_j^\dagger v\|
\leq
\|(I+\hat n)^{1/2}v\|.
$$
The same argument with the weight $(I+\hat n)^{s/2}$ inserted gives
$$
\|\hat a_j v\|_{s,\mathrm{osc}}
+
\|\hat a_j^\dagger v\|_{s,\mathrm{osc}}
\leq
C_s\|v\|_{s+1,\mathrm{osc}}.
$$
Iterating this estimate over a product of at most $d$ ladder operators gives Eq.~\eqref{eq:polynomial-control-number}. A finite linear combination only changes the constant.
\end{proof}

\subsection{Finite-time Galerkin error}
\label{app:finite-time-galerkin-proof}

We prove \Cref{thm:finite-time-galerkin-truncation}. 
Let $L=P(\hat q,\hat p)$ be an ordered polynomial of total degree $d$. 
By \Cref{lem:polynomial-control-number},
\begin{equation}
\label{eq:polynomial-control-number-unweighted}
\|Lv\|
\leq
C_P\|v\|_{d,\mathrm{osc}}
\end{equation}
on the Hermite span. More generally, with oscillator weights inserted,
\begin{equation}
\label{eq:weighted-polynomial-control-number}
\|Lv\|_{s,\mathrm{osc}}
\leq
C_{P,s}\|v\|_{s+d,\mathrm{osc}}.
\end{equation}
This follows from the fact that each $\hat q_j$ or $\hat p_j$ is a linear combination of $\hat a_j$ and $\hat a_j^\dagger$, and hence changes total Hermite degree by at most one.

Assume stability, that is, the projected dynamics satisfy
\begin{equation}
\label{eq:projected-stability}
\|e^{tL_m}\|\leq C_T,
\qquad
0\leq t\leq T,
\end{equation}
with $C_T$ independent of $m$. 
Define the error $e_m$ between the solution to the dynamics in the projected space $u_m$ and the solution to the dynamics in the full space projected on $\mathcal{H}_{\leq m}$ to be
$$
e_m(t)=\Pi_m u(t)-u_m(t).
$$
Since $u$ solves $\partial_t u=Lu$ and $u_m$ solves $\partial_t u_m=L_m u_m$, we have that the error follows the same ODE as $u_m$, plus an additional term of high occupation numbers components of $u$ feeding into $\mathcal{H}_{\leq m}$. 
\begin{align}
\partial_t e_m(t)
&=
\Pi_mLu(t)-L_m u_m(t) \nonumber\\
&=
L_m e_m(t)+\Pi_mL(I-\Pi_m)u(t).
\end{align}
Moreover $e_m(0)=0$. Duhamel's formula gives
\begin{equation}
\label{eq:duhamel-galerkin-error}
e_m(t)
=
\int_0^t
e^{(t-\tau)L_m}
\Pi_mL(I-\Pi_m)u(\tau)\,d\tau .
\end{equation}
Therefore,
\begin{equation}
\label{eq:em-bound-before-tail}
\|e_m(t)\|
\leq
C_T
\int_0^t
\|L(I-\Pi_m)u(\tau)\|\,d\tau .
\end{equation}
Using \Cref{eq:polynomial-control-number-unweighted},
$$
\|L(I-\Pi_m)u(\tau)\|
\leq
C_P\|(I-\Pi_m)u(\tau)\|_{d,\mathrm{osc}}.
$$
Since $\Pi_m$ commutes with $\hat n$,
$$
\|(I-\Pi_m)u(\tau)\|_{d,\mathrm{osc}}
=
\|(I-\Pi_m)(I+\hat n)^{d/2}u(\tau)\|.
$$
Applying \Cref{prop:hermite-tail-bound} to $(I+\hat n)^{d/2}u(\tau)$ gives
$$
\|(I-\Pi_m)u(\tau)\|_{d,\mathrm{osc}}
\leq
(m+1)^{-s/2}
\|u(\tau)\|_{s+d,\mathrm{osc}}.
$$
Hence
\begin{equation}
\label{eq:em-bound-after-tail}
\|e_m(t)\|
\leq
C_T C_P t
(m+1)^{-s/2} C_{s+d}
\end{equation}
Finally,
$$
u(t)-u_m(t)
=
(I-\Pi_m)u(t)+e_m(t).
$$
Using \Cref{prop:hermite-tail-bound} once more,
\begin{equation}
\label{eq:finite-time-galerkin-detailed}
\|u(t)-u_m(t)\|
\leq
(m+1)^{-s/2}
\left(
\|u(t)\|_{s,\mathrm{osc}}
+
C_T C_P t
\sup_{0\leq \tau\leq T}
\|u(\tau)\|_{s+d,\mathrm{osc}}
\right).
\end{equation}
Since $\|u(t)\|_{s,\mathrm{osc}}\le \|u(t)\|_{s+d,\mathrm{osc}}$, the same constant $C_s$ controls both terms.

The assumption in \Cref{eq:projected-stability} and the boundedness of the oscillator Sobolev norms are not automatic for arbitrary polynomial expressions $P(\hat q,\hat p)$. 
Polynomial bosonic operators may generate flows whose high-order oscillator norms grow rapidly, or even develop finite-time singularities at the level of the associated classical or Heisenberg dynamics. 
For example, for
$$
H=\frac{1}{2}(\hat p\hat q^2+\hat q^2\hat p),
$$
the Heisenberg equation gives
$$
\frac{d}{dt}\hat q(t)
=
i[H,\hat q(t)]
=
\hat q(t)^2.
$$
The corresponding classical equation $\dot q=q^2$ has finite-time blow-up. This does not contradict unitarity of $e^{-itH}$ when $H$ is self-adjoint; it says that polynomial observables and oscillator Sobolev norms need not remain uniformly controlled. This is why \Cref{thm:finite-time-galerkin-truncation} is stated with explicit stability and regularity assumptions.

\subsection{Ground-state consistency}
\label{app:ground-state-consistency}

We prove Theorem~\ref{thm:statics-convergence}. Let
$$
\delta_m:=\|(I-\Pi_m)\phi\|.
$$
By Proposition~\ref{prop:hermite-tail-bound} and the assumption on $\phi$, $\delta_m\to 0$. Hence, for all sufficiently large $m$,
$\Pi_m\phi\neq 0$ and $\|\Pi_m\phi\|\ge 1/2$. Define the normalized trial state
$$
\widetilde \phi_m
:=
\frac{\Pi_m\phi}{\|\Pi_m\phi\|}.
$$
By the variational principle,
\begin{equation}
0\le \lambda_{0,m}
\le
\langle \widetilde\phi_m,H\widetilde\phi_m\rangle .
\label{eq:lambda0-variational}
\end{equation}
Since $H\phi=0$, we have
$$
H\Pi_m\phi
=
-H(I-\Pi_m)\phi .
$$
Therefore
\begin{align}
\langle \widetilde\phi_m,H\widetilde\phi_m\rangle
&=
\frac{\langle \Pi_m\phi,H\Pi_m\phi\rangle}{\|\Pi_m\phi\|^2}  \\
&\le
\frac{\|H(I-\Pi_m)\phi\|}{\|\Pi_m\phi\|}.
\label{eq:trial-energy-bound}
\end{align}
Using the oscillator-control assumption \eqref{eq:H-osc-control},
$$
\|H(I-\Pi_m)\phi\|
\le
C_H\|(I-\Pi_m)\phi\|_{d,\mathrm{osc}} .
$$
Since $\Pi_m$ commutes with $\hat n$,
$$
\|(I-\Pi_m)\phi\|_{d,\mathrm{osc}}
=
\|(I-\Pi_m)(I+\hat n)^{d/2}\phi\|.
$$
Applying Proposition~\ref{prop:hermite-tail-bound} to $(I+\hat n)^{d/2}\phi$ gives
\begin{equation}
\|(I-\Pi_m)\phi\|_{d,\mathrm{osc}}
\le
(m+1)^{-s/2}\|\phi\|_{s+d,\mathrm{osc}} .
\label{eq:tail-d-osc-bound}
\end{equation}
For all sufficiently large $m$, $\|\Pi_m\phi\|\ge 1/2$. Combining the previous estimates gives
\begin{equation}
0\le \lambda_{0,m}
\le
2C_H\|\phi\|_{s+d,\mathrm{osc}}(m+1)^{-s/2}
=
E_m .
\label{eq:lambda0-Em}
\end{equation}

Now let $\phi_m$ be a normalized ground state of $H_m$. Since $\phi_m\in\mathcal H_{\le m}$,
$$
\langle \phi_m,H\phi_m\rangle
=
\langle \phi_m,H_m\phi_m\rangle
=
\lambda_{0,m}.
$$
Decompose the ground state of the truncated operator along the ground state of the infinite-dimensional operator
$$
\phi_m=a_m\phi+\eta_m,
\qquad
\eta_m\perp \phi .
$$
The spectral gap of $H$ implies
\begin{equation}
\lambda_{0,m}
=
\langle \phi_m,H\phi_m\rangle
\ge
\gamma\|\eta_m\|^2 .
\label{eq:gap-energy-state-error}
\end{equation}
Using \eqref{eq:lambda0-Em}, we obtain
\begin{equation}
\|(I-|\phi\rangle\langle\phi|)\phi_m\|
=
\|\eta_m\|
\le
\left(\frac{E_m}{\gamma}\right)^{1/2}.
\label{eq:projection-error-proof}
\end{equation}
After choosing the phase of $\phi_m$ so that $a_m\ge 0$, we have
$$
\|\phi_m-\phi\|^2
=
2(1-a_m)
\le
2\|\eta_m\|^2,
$$
and therefore
\begin{equation}
\|\phi_m-\phi\|
\le
\left(\frac{2E_m}{\gamma}\right)^{1/2}.
\label{eq:vector-error-proof}
\end{equation}

It remains to prove the lower bound on the truncated gap. Let $v\in\mathcal H_{\le m}$ be normalized
and orthogonal to $\phi_m$. Since $v\perp \phi_m$,
$$
|\langle v,\phi\rangle|
=
|\langle v,\phi-\phi_m\rangle|
\le
\|\phi-\phi_m\|
\le
\left(\frac{2E_m}{\gamma}\right)^{1/2}.
$$
Using the spectral gap of $H$ again,
\begin{align}
\langle v,H_mv\rangle
&=
\langle v,Hv\rangle  \\
&\ge
\gamma\left(1-|\langle v,\phi\rangle|^2\right)  \\
&\ge
\gamma-2E_m .
\end{align}
By the min-max principle, this implies
\begin{equation}
\lambda_{1,m}\ge \gamma-2E_m .
\label{eq:lambda1m-bound}
\end{equation}
Together with $\lambda_{0,m}\le E_m$, this gives
\begin{equation}
\gamma_m
=
\lambda_{1,m}-\lambda_{0,m}
\ge
\gamma-3E_m .
\end{equation}
This proves the theorem.

\section{Details for the symmetric bosonic encoding}
\label{app:symmetric-bosonic-encoding}

\subsection{The symmetric encoding isometry}
\label{app:symmetric-encoding-isometry}

We prove \Cref{prop:symmetric-encoding-global-cut-off}. Let
$$
\mathcal K_N
=
\operatorname{span}
\{
\ket{\emptyset},\ket{1},\ldots,\ket{N}
\}.
$$
For $\alpha\in\mathbb N^N$ with $|\alpha|\leq m$, define the multiset
$$
M_\alpha
=
\{
\underbrace{1,\ldots,1}_{\alpha_1},
\underbrace{2,\ldots,2}_{\alpha_2},
\ldots,
\underbrace{N,\ldots,N}_{\alpha_N},
\underbrace{\emptyset,\ldots,\emptyset}_{m-|\alpha|}
\}.
$$
Let $\mathcal P_\alpha$ be the set of distinct ordered strings obtained by permuting this multiset. Its cardinality is
\begin{equation}
\label{eq:number-distinct-permutations}
|\mathcal P_\alpha|
=
\frac{m!}{(m-|\alpha|)!\prod_{j=1}^N \alpha_j!}.
\end{equation}
The encoded state is
\begin{equation}
\label{eq:encoded-alpha-state}
\ket{\alpha}_F
\mapsto
\frac{1}{\sqrt{|\mathcal P_\alpha|}}
\sum_{b\in \mathcal P_\alpha}
\ket{b_1,\ldots,b_m}_Q=:\ket{\alpha;m}_Q.
\end{equation}
This state is invariant under permutations of the $m$ registers, hence belongs to $\operatorname{Sym}^m(\mathcal K_N)$.

If $\alpha\neq \beta$, then the two multisets $M_\alpha$ and $M_\beta$ are different. Therefore $\mathcal P_\alpha$ and $\mathcal P_\beta$ are disjoint sets of computational-basis strings. It follows that
$$
\langle \alpha|\beta\rangle_Q=0.
$$
For $\alpha=\beta$, \Cref{eq:encoded-alpha-state} is normalized by construction. Hence the family $\{\ket{\alpha;m}_Q\}_{|\alpha|\leq m}$ is orthonormal, and the map $V_m\ket{\alpha}_F=\ket{\alpha;m}_Q$ is an isometry for $|\alpha|\leq m$.

It remains to identify the image. A basis vector of $\operatorname{Sym}^m(\mathcal K_N)$ is specified by occupation numbers of the $N+1$ labels
$$
(\alpha_1,\ldots,\alpha_N,\alpha_\emptyset),
\qquad
\alpha_\emptyset+\sum_{j=1}^N\alpha_j=m.
$$
Equivalently, it is specified by an $N$-mode occupation vector $\alpha=(\alpha_1,\ldots,\alpha_N)$ satisfying $|\alpha|\leq m$, with $\alpha_\emptyset=m-|\alpha|$. Thus the symmetric basis vectors are in one-to-one correspondence with the global Hermite modes retained in $\mathcal H_{\leq m}$. Therefore $V_m$ maps $\mathcal H_{\leq m}$ onto $\operatorname{Sym}^m(\mathcal K_N)$.

The dimension is the number of weak compositions of $m$ into $N+1$ parts:
\begin{equation}
\label{eq:symmetric-dimension-proof}
\dim \operatorname{Sym}^m(\mathcal K_N)
=
\binom{(N+1)+m-1}{m}
=
\binom{N+m}{m}.
\end{equation}
This is also the number of multi-indices $\alpha\in\mathbb N^N$ with $|\alpha|\leq m$.

\subsection{Encoded creation and annihilation operators}
\label{app:encoded-bosonic-primitives}

We prove \Cref{prop:encoded-bosonic-primitives}. Recall that
$$
\mathcal K_N
=
\operatorname{span}\{\ket{\emptyset},\ket{1},\ldots,\ket{N}\}.
$$
For register $r\in\{1,\ldots,m\}$, define the single-register transitions
$$
b_{j,r}^\dagger
=
\ket{j}\!\bra{\emptyset}_r,
\qquad
b_{j,r}
=
\ket{\emptyset}\!\bra{j}_r.
$$
Let
\begin{equation}
\label{eq:unscaled-register-ladder}
B_j^\dagger
=
\sum_{r=1}^m b_{j,r}^\dagger,
\qquad
B_j
=
\sum_{r=1}^m b_{j,r}.
\end{equation}
These operators act symmetrically on the $m$ registers and preserve the symmetric subspace. Let
\begin{equation}
\label{eq:empty-number-operator}
N_\emptyset
=
\sum_{r=1}^m
\ket{\emptyset}\!\bra{\emptyset}_r.
\end{equation}
On the encoded state $\ket{\alpha}_Q$, one has
$$
N_\emptyset\ket{\alpha}_Q
=
(m-|\alpha|)\ket{\alpha}_Q.
$$
The unnormalized transition operators have the actions
\begin{equation}
\label{eq:B-dagger-action}
B_j^\dagger\ket{\alpha}_Q
=
\sqrt{(\alpha_j+1)(m-|\alpha|)}
\ket{\alpha+e_j}_Q
\end{equation}
for $|\alpha|<m$, and zero for $|\alpha|=m$. Similarly,
\begin{equation}
\label{eq:B-action}
B_j\ket{\alpha}_Q
=
\sqrt{\alpha_j(m-|\alpha|+1)}
\ket{\alpha-e_j}_Q.
\end{equation}
These are the standard matrix elements of occupation operators on the symmetric power. The factor $\sqrt{\alpha_j+1}$ or $\sqrt{\alpha_j}$ comes from the occupation of mode $j$, while the factor involving $m-|\alpha|$ comes from the occupation of the empty label.
To remove the empty-label normalization, define
\begin{equation}
\label{eq:encoded-ladder-from-B}
A_j^\dagger
=
B_j^\dagger N_\emptyset^{-1/2},
\qquad
A_j
=
B_j (N_\emptyset+I)^{-1/2},
\end{equation}
where $N_\emptyset^{-1/2}$ is defined to be zero on the kernel of $N_\emptyset$. Acting on $\ket{\alpha}_Q$, \Cref{eq:B-dagger-action,eq:B-action} give
$$
A_j^\dagger\ket{\alpha}_Q
=
\begin{cases}
\sqrt{\alpha_j+1}\ket{\alpha+e_j}_Q, & |\alpha|<m,\\
0, & |\alpha|=m,
\end{cases}
$$
and
$$
A_j\ket{\alpha}_Q
=
\sqrt{\alpha_j}\ket{\alpha-e_j}_Q.
$$
Therefore $A_j$ and $A_j^\dagger$ are exactly the encoded versions of $\Pi_m\hat a_j\Pi_m$ and $\Pi_m\hat a_j^\dagger\Pi_m$.
The number operator is
\begin{equation}
\label{eq:encoded-number-register}
N_m
=
mI-N_\emptyset.
\end{equation}
Hence
$$
N_m\ket{\alpha}_Q=|\alpha|\ket{\alpha}_Q,
$$
which agrees with the total Hermite number operator on $\mathcal H_{\leq m}$. The formulas for $Q_j$ and $P_j$ follow from the identities
$$
\hat q_j=\frac{\hat a_j+\hat a_j^\dagger}{\sqrt 2},
\qquad
\hat p_j=\frac{\hat a_j-\hat a_j^\dagger}{i\sqrt 2}.
$$
We finally comment on block encodings. 
The operators $B_j$ and $B_j^\dagger$ are sums of $m$ identical single-register transitions, selected by a register index $r$. 
They can therefore be block encoded by a standard linear-combination-of-unitaries construction using a uniform superposition over $r\in\{1,\ldots,m\}$. 
The diagonal corrections $N_\emptyset^{-1/2}$ and $(N_\emptyset+I)^{-1/2}$ depend only on the number of empty slots and can be implemented by reversible counting and arithmetic. 
This gives block encodings of $A_j$ and $A_j^\dagger$ with normalization $O(\sqrt m)$. The operators $Q_j$ and $P_j$ are fixed linear combinations of these, and $N_m$ is diagonal with spectrum contained in $\{0,\ldots,m\}$, giving normalization $O(m)$.

A linear form can be constructed using Linear Combination of Unitary (LCU)
$$
\mathcal L(v,w)
=
\sum_j v_jQ_j+\sum_j w_jP_j,
$$
We use state preparation for the coefficient vectors $v$ and $w$ together with the same slot-selection construction. The normalization is $O(\sqrt m(\|v\|+\|w\|))$.

\subsection{Boundary-safe polynomial evaluation and polynomial block encodings}
\label{app:block-encoding-polynomial-operators}

We give the proof of \Cref{thm:structured-polynomial-block-encoding}. 
Throughout this subsection, the degree $d$ is the constant degree appearing in \Cref{def:structured-operator-polynomial}.
There are two ingredients: boundary-safe evaluation and standard composition rules for block encodings.

\subsubsection{Boundary-safe evaluation}
\label{app:boundary-safe-polynomial-evaluation}
There is one technical point. If $P$ has degree $d$, then applying $P(\hat q,\hat p)$ to a state in $\mathcal H_{\leq m}$ can pass through Hermite levels as high as $m+d$ before returning to $\mathcal H_{\leq m}$. Therefore, the correct Galerkin matrix is not obtained by simply replacing $\hat q_j,\hat p_j$ with their $m$-truncated versions inside the polynomial. One must evaluate the polynomial in a slightly enlarged cut-off. Explicitly, if $M=m+d$, then
\begin{equation}
\label{eq:boundary-safe-polynomial-evaluation}
\Pi_m P(\hat q,\hat p)\Pi_m
=
\Pi_m P(\Pi_M\hat q\Pi_M,\Pi_M\hat p\Pi_M)\Pi_m .
\end{equation}
We refer to this as boundary-safe polynomial evaluation.

Let $P(\hat q,\hat p)$ be an ordered polynomial of degree at most $d$. Set $M=m+d$. We claim that
\begin{equation}
\label{eq:boundary-safe-proof-identity}
\Pi_m P(\hat q,\hat p)\Pi_m
=
\Pi_m P(\Pi_M\hat q\Pi_M,\Pi_M\hat p\Pi_M)\Pi_m .
\end{equation}
It is enough to prove the claim for an ordered monomial
$$
X_1X_2\cdots X_k,
\qquad
k\leq d,
$$
where each $X_\ell$ is one of the operators $\hat q_j$ or $\hat p_j$. 
Each $X_\ell$ changes the total Hermite degree by at most one. Therefore, starting from a vector in $\mathcal H_{\leq m}$, every intermediate vector appearing after applying at most $k$ factors lies in $\mathcal H_{\leq m+k}\subseteq \mathcal H_{\leq M}$.

It follows that inserting $\Pi_M$ between the factors has no effect on the component that starts and ends in $\mathcal H_{\leq m}$:
$$
\Pi_m X_1\cdots X_k\Pi_m
=
\Pi_m
(\Pi_M X_1\Pi_M)
\cdots
(\Pi_M X_k\Pi_M)
\Pi_m .
$$
Linearity gives \Cref{eq:boundary-safe-proof-identity} for any ordered polynomial $P$.

\subsubsection{Block encoding of structured linear forms}
\label{app:block-encoding-linear-forms}

Let
$$
\mathcal L(a,v,w)
=
aI+\sum_{j=1}^N v_jQ_j+\sum_{j=1}^N w_jP_j
$$
on the cut-off-$M$ encoded space. By \Cref{prop:encoded-bosonic-primitives}, each $Q_j$ and $P_j$ admits a block encoding with normalization $O(\sqrt M)$. If the normalized states
$$
\ket{v}
=
\frac{1}{\|v\|}
\sum_{j=1}^N v_j\ket j,
\qquad
\ket{w}
=
\frac{1}{\|w\|}
\sum_{j=1}^N w_j\ket j
$$
are efficiently preparable, a standard linear-combination-of-unitaries~\cite{childs2012hamiltonian-554} construction gives block encodings of
$$
\sum_j v_jQ_j
\qquad\text{and}\qquad
\sum_j w_jP_j
$$
with normalizations $O(\sqrt M\|v\|)$ and $O(\sqrt M\|w\|)$, respectively. Including the scalar term $aI$ gives a block encoding of $\mathcal L(a,v,w)$ with normalization
\begin{equation}
\label{eq:linear-form-normalization}
\alpha(a,v,w)
=
O\left(
|a|+\sqrt M(\|v\|+\|w\|)
\right).
\end{equation}

\subsubsection{Composition into ordered polynomials}
\label{app:block-encoding-polynomial-composition}

Consider one ordered monomial
$$
P_r(Q,P)
=
c_r
\prod_{\ell=1}^{k_r}
\mathcal L(a_{r,\ell},v_{r,\ell},w_{r,\ell}).
$$
For each factor, let $\alpha_{r,\ell}$ denote the normalization in \Cref{eq:linear-form-normalization}. Composing block encodings gives a block encoding of the ordered product with normalization
\begin{equation}
\label{eq:monomial-normalization}
\alpha_r
=
|c_r|
\prod_{\ell=1}^{k_r}\alpha_{r,\ell}.
\end{equation}
The ordering is preserved by applying the factor block encodings in the same order as in the monomial.

Finally, summing over $r=1,\ldots,R$ by a linear-combination-of-unitaries construction gives a block encoding of
$$
P(Q,P)
=
\sum_{r=1}^R P_r(Q,P)
$$
with normalization
\begin{equation}
\label{eq:polynomial-normalization-appendix}
\alpha_P
=
\sum_{r=1}^R
|c_r|
\prod_{\ell=1}^{k_r}
O\left(
|a_{r,\ell}|
+
\sqrt M
(\|v_{r,\ell}\|+\|w_{r,\ell}\|)
\right).
\end{equation}
Since $M=m+d$, this is polynomial under the conditions of \Cref{thm:structured-polynomial-block-encoding}.

Combining this construction with boundary-safe evaluation \Cref{eq:boundary-safe-proof-identity} yields a block encoding of
$$
\Pi_m P(\hat q,\hat p)\Pi_m
$$
on the cut-off-$m$ subspace. This proves \Cref{thm:structured-polynomial-block-encoding}.

\subsection{Readout from the symmetric encoding}
\label{app:readout-observables-marginals}

We justify \Cref{prop:accessible-readout}. Let
$$
\ket{\psi_m}_Q
=
\sum_{|\alpha|\leq m}c_\alpha\ket{\alpha}_Q.
$$

\subsubsection{Block-encoded observables}

If $O_m$ has an $(\alpha_O,a,\epsilon)$ block encoding $U_O$, then the expectation
$$
\langle O_m\rangle_{\psi_m}
=
\bra{\psi_m}O_m\ket{\psi_m}
$$
can be estimated by a Hadamard test or amplitude estimation applied to $U_O$ and the state-preparation circuit for $\ket{\psi_m}_Q$. The sample complexity scales polynomially in $\alpha_O$ and the desired inverse precision. When $O_m=P(Q,P)$ is a structured polynomial observable, \Cref{thm:structured-polynomial-block-encoding} supplies the required block encoding.

\subsubsection{Occupation diagnostics}

The total occupation operator is diagonal in the encoded basis:
$$
N_m\ket{\alpha}_Q=|\alpha|\ket{\alpha}_Q.
$$
Equivalently, in the register encoding,
$$
N_m=mI-N_\emptyset,
$$
where $N_\emptyset$ counts the number of empty labels. Thus $N_m$ can be measured by counting the number of registers not equal to $\emptyset$. This gives direct access to the distribution
$$
\Pr[K=k]
=
\sum_{|\alpha|=k}|c_\alpha|^2,
\qquad
0\leq k\leq m.
$$
Moments such as $\langle N_m^r\rangle$ are then obtained classically from samples of $K$.

Large mass near $K=m$ indicates that the state is interacting strongly with the boundary of the cut-off. This does not by itself prove that the approximation is bad, but it is a useful diagnostic that the cut-off should be increased.

\subsubsection{Reduced Hermite matrices and marginal sampling}
\label{app:reduced-hermite-marginal-sampling}
Sampling $x\sim\rho_{\infty}$ cannot be efficient as $x\in\mathbb R^N$, and simply writing a single sample has exponential complexity when $N=2^n$. 
A meaningful target is constant-dimensional marginal sampling. 
We describe the one-coordinate case; the extension to any constant-size set of coordinates is analogous.
Consider an encoded state, and its corresponding Hermite expansion
$$
\ket{\psi_m}_Q
=
\sum_{|\alpha|\leq m} c_\alpha \ket{\alpha}_Q,
\qquad 
\psi_m(x)
=
\sum_{|\alpha|\leq m} c_\alpha h_\alpha(x)
$$
Fix a coordinate of interest $j\in\{1,\ldots,N\}$ and split the multi-index as
$$
\alpha=(\alpha_{\neq j},\alpha_j).
$$
Define the reduced Hermite matrix of coordinate $j$ by
\begin{equation}
\label{eq:reduced-hermite-matrix-single-coordinate}
\Gamma^{(j)}_{k\ell}
=
\sum_{\alpha_{\neq j}}
c_{\alpha_{\neq j},k}\overline{c_{\alpha_{\neq j},\ell}},
\qquad
0\leq k,\ell\leq m.
\end{equation}
Here coefficients outside the global cut-off are understood to be zero. Equivalently, the sum only includes indices for which both $|\alpha_{\neq j}|+k\leq m$ and $|\alpha_{\neq j}|+\ell\leq m$.

By orthonormality of the Hermite basis in the unobserved variables, the one-coordinate marginal density is
\begin{equation}
\label{eq:one-coordinate-marginal-hermite}
\rho_j(x_j)
=
\int_{\mathbb R^{N-1}}
|\psi_m(x)|^2\,dx_{\neq j}
=
\sum_{k,\ell=0}^m
\Gamma^{(j)}_{k\ell}h_k(x_j)h_\ell(x_j).
\end{equation}
Thus, once $\Gamma^{(j)}$ is known, the marginal is a one-dimensional density with $m+1$ Hermite modes. It can be evaluated classically in time polynomial in $m$. 
Standard one-dimensional quadrature and inverse-CDF sampling then give classical marginal samples with cost polynomial in $m$ and the desired sampling accuracy.

It remains to explain how the entries of $\Gamma^{(j)}$ are accessed from the encoded state. Define the mode-$j$ occupation counter
\begin{equation}
\label{eq:mode-j-counter}
N_j
=
\sum_{t=1}^m
\ket{j}\!\bra{j}^{(t)}.
\end{equation}
On the symmetric logical space, this agrees with $A_j^\dagger A_j$, and
$$
N_j\ket{\alpha}_Q=\alpha_j\ket{\alpha}_Q.
$$
Let $P^{(j)}_0$ be the projector onto zero occupation of mode $j$. It can be implemented by reversibly computing $N_j$ and checking whether the result is zero. Algebraically, on the truncated space, it is the polynomial
\begin{equation}
\label{eq:zero-occupation-projector-mode-j}
P^{(j)}_0
=
\prod_{r=1}^m
\left(
I-\frac{N_j}{r}
\right).
\end{equation}
Indeed, on a basis state with $j$-occupation $a_j\in\{0,\ldots,m\}$, this polynomial evaluates to $1$ if $a_j=0$ and to $0$ otherwise.

For $0\leq k,\ell\leq m$, define
\begin{equation}
\label{eq:mode-j-matrix-unit}
M^{(j)}_{k\ell}
=
\frac{(A_j^\dagger)^\ell}{\sqrt{\ell!}}\,
P^{(j)}_0\,
\frac{A_j^k}{\sqrt{k!}}.
\end{equation}
This operator removes $k$ particles in mode $j$, checks that the mode is then empty, and creates $\ell$ particles in mode $j$. More precisely,
\begin{equation}
\label{eq:mode-j-matrix-unit-action}
M^{(j)}_{k\ell}
\ket{\alpha_{\neq j},k}_Q
=
\ket{\alpha_{\neq j},\ell}_Q,
\end{equation}
whenever both states lie inside the global cut-off, and $M^{(j)}_{k\ell}$ annihilates basis states whose $j$-occupation is not $k$. Therefore
\begin{equation}
\label{eq:mode-j-matrix-unit-expansion}
M^{(j)}_{k\ell}
=
\sum_{\alpha_{\neq j}}
\ket{\alpha_{\neq j},\ell}_Q
\bra{\alpha_{\neq j},k}_Q,
\end{equation}
where the sum is restricted to indices for which both states belong to $\mathcal H_{\leq m}$.

Taking the expectation in $\ket{\psi_m}_Q$ gives
\begin{equation}
\label{eq:reduced-hermite-entry-expectation}
\Gamma^{(j)}_{k\ell}
=
\bra{\psi_m}M^{(j)}_{k\ell}\ket{\psi_m}_Q.
\end{equation}
For $k\neq \ell$, the operator $M^{(j)}_{k\ell}$ is not Hermitian. Since
$$
\left(M^{(j)}_{k\ell}\right)^\dagger
=
M^{(j)}_{\ell k},
$$
its real and imaginary parts are obtained from the Hermitian observables
\begin{equation}
\label{eq:matrix-unit-hermitian-parts}
\frac{1}{2}
\left(
M^{(j)}_{k\ell}+M^{(j)}_{\ell k}
\right),
\qquad
\frac{1}{2i}
\left(
M^{(j)}_{k\ell}-M^{(j)}_{\ell k}
\right).
\end{equation}

The operators $M^{(j)}_{k\ell}$ are built from the encoded primitives $A_j$, $A_j^\dagger$, and $N_j$. Since $k,\ell\leq m$, and $m$ is the global truncation parameter, each $M^{(j)}_{k\ell}$ has a block encoding with normalization and gate complexity polynomial in $m$ and $\log N$. Hence each entry of $\Gamma^{(j)}$ can be estimated with polynomial query and gate complexity.

More explicitly, suppose each $M^{(j)}_{k\ell}$ is block encoded with normalization $\eta_M=\operatorname{poly}(m)$ and gate cost $C_M=\operatorname{poly}(m,\log N)$. Then estimating all entries of $\Gamma^{(j)}$ to entrywise accuracy $\delta$ and failure probability $p_{\mathrm{fail}}$ has sampling query complexity
\begin{equation}
\label{eq:marginal-entry-sampling-complexity}
O\left(
(m+1)^2
\frac{\eta_M^2}{\delta^2}
\log\frac{(m+1)^2}{p_{\mathrm{fail}}}
\right),
\end{equation}
or, with coherent amplitude estimation,
\begin{equation}
\label{eq:marginal-entry-ae-complexity}
O\left(
(m+1)^2
\frac{\eta_M}{\delta}
\log\frac{(m+1)^2}{p_{\mathrm{fail}}}
\right).
\end{equation}
The corresponding gate complexity is obtained by multiplying by $C_M$.

Finally, entrywise accuracy controls the reconstructed marginal. Suppose
$$
|\widehat\Gamma^{(j)}_{k\ell}-\Gamma^{(j)}_{k\ell}|\leq \delta
$$
for all $0\leq k,\ell\leq m$, and define
$$
\widehat\rho_j(x_j)
=
\sum_{k,\ell=0}^m
\widehat\Gamma^{(j)}_{k\ell}h_k(x_j)h_\ell(x_j).
$$
Then
\begin{align}
\label{eq:marginal-l1-error}
\|\widehat\rho_j-\rho_j\|_{L^1(\mathbb R)}
&\leq
\sum_{k,\ell=0}^m
|\widehat\Gamma^{(j)}_{k\ell}-\Gamma^{(j)}_{k\ell}|
\|h_kh_\ell\|_{L^1(\mathbb R)}
\nonumber\\
&\leq
\sum_{k,\ell=0}^m
\delta\|h_k\|_{L^2(\mathbb R)}\|h_\ell\|_{L^2(\mathbb R)}
\nonumber\\
&=
(m+1)^2\delta .
\end{align}
Thus choosing
$$
\delta=O\left(\frac{\epsilon_{\mathrm{marg}}}{(m+1)^2}\right)
$$
gives an $L^1$-accurate reconstruction of the marginal density. The final sampling step is then classical and one-dimensional.

It is useful to contrast this route with the quantum Hermite transform. In a local tensor-product truncation, one stores a coefficient vector in a register of the form
$$
\ket{\psi}
=
\sum_{\alpha_1,\ldots,\alpha_N}
c_\alpha
\ket{\alpha_1}\cdots\ket{\alpha_N}.
$$
An inverse quantum Hermite transform as in~\cite{jain2026efficient-da2} applied coordinatewise would map this coefficient representation to a discretized position representation,
\begin{equation}
\label{eq:inverse-qht-local-truncation}
\ket{\psi}
\longmapsto
\sum_x \psi(x)\ket{x},
\end{equation}
after which measurement would sample from the discretized density. 
This is the natural strategy in a tensor-product encoding.

The global symmetric encoding is different. It stores the state by occupations of indistinguishable excitations, not by assigning a Hermite register to each coordinate. There is therefore no immediate coordinatewise tensor-product structure on which to apply $(\mathrm{QHT}^{-1})^{\otimes N}$. Moreover, when $N=2^n$, a full position-basis sample has exponentially many coordinates. For the global truncation, the reduced-Hermite-matrix construction above is the appropriate readout primitive: it extracts low-dimensional marginals without expanding the whole state into a tensor-product position representation.

\section{Fokker--Planck and Inverse problems}
\subsection{Stationarity of Gibbs densities for Fokker--Planck equations}
\label{app:gibbs-stationarity-fokker-planck}

We justify the stationarity claims used in \Cref{subsec:general-fokker-planck}. Consider
\begin{equation}
\label{eq:appendix-general-fp}
\partial_t\rho
=
\nabla\cdot(\rho F)+\Delta\rho .
\end{equation}
Let $V:\mathbb R^N\to\mathbb R$ be such that
$$
Z=\int_{\mathbb R^N}e^{-V(x)}dx<+\infty,
$$
and define
$$
\rho_\infty(x)=Z^{-1}e^{-V(x)}.
$$
The density $\rho_\infty$ is stationary for \Cref{eq:appendix-general-fp} if and only if
\begin{equation}
\label{eq:appendix-stationarity-condition}
\nabla\cdot(\rho_\infty F)+\Delta\rho_\infty=0.
\end{equation}
Assume that the drift admits the decomposition
\begin{equation}
\label{eq:appendix-drift-decomposition}
F=\nabla V+G.
\end{equation}
Since $\rho_\infty=Z^{-1}e^{-V}$, the constant $Z^{-1}$ plays no role in the stationarity equation. We compute
\begin{align}
\nabla\cdot(e^{-V}F)+\Delta e^{-V}
&=
\nabla\cdot(e^{-V}\nabla V)
+
\nabla\cdot(e^{-V}G)
+
\Delta e^{-V}.
\end{align}
Using
$$
\nabla e^{-V}=-e^{-V}\nabla V,
$$
we have
$$
\Delta e^{-V}
=
\nabla\cdot(-e^{-V}\nabla V).
$$
Therefore
$$
\nabla\cdot(e^{-V}\nabla V)+\Delta e^{-V}=0,
$$
and hence
\begin{equation}
\label{eq:appendix-stationarity-reduction}
\nabla\cdot(e^{-V}F)+\Delta e^{-V}
=
\nabla\cdot(e^{-V}G).
\end{equation}
Thus $\rho_\infty=Z^{-1}e^{-V}$ is stationary if and only if
\begin{equation}
\label{eq:appendix-weighted-div-free}
\nabla\cdot(e^{-V}G)=0.
\end{equation}
Equivalently,
\begin{equation}
\label{eq:appendix-weighted-div-free-expanded}
\nabla\cdot G=G\cdot\nabla V.
\end{equation}
This is the weighted divergence-free condition. It is the ordinary divergence-free condition only when $V$ is constant.

\subsection{Ground-state transform with weighted divergence-free drift}
\label{app:ground-state-transform-fokker-planck}

We derive \Cref{eq:ground-state-transformed-equation}. Consider the Fokker--Planck equation
\begin{equation}
\label{eq:app-general-nonreversible-fp}
\partial_t\rho
=
\nabla\cdot(\rho(\nabla V+G))+\Delta\rho,
\qquad
\nabla\cdot(e^{-V}G)=0.
\end{equation}
Let
$$
\phi=e^{-V/2}
$$
up to an irrelevant normalization constant, and write
$$
\rho=\phi\psi.
$$
Since $\phi$ is time independent,
$$
\partial_t\rho=\phi\,\partial_t\psi.
$$
We split the transformed generator into the reversible contribution generated by $\nabla V$ and the residual contribution generated by $G$.

First consider the reversible part. We compute
$$
\nabla(\phi\psi)
=
\phi\nabla\psi+\psi\nabla\phi.
$$
Since
$$
\nabla\phi
=
-\frac12\phi\nabla V,
$$
we have
$$
\nabla(\phi\psi)
=
\phi\nabla\psi-\frac12\phi\psi\nabla V.
$$
Therefore
\begin{align}
\nabla\cdot(\rho\nabla V)+\Delta\rho
&=
\nabla\cdot(\phi\psi\nabla V)+\Delta(\phi\psi) \nonumber\\
&=
\nabla\cdot
\left(
\phi\nabla\psi+\frac12\phi\psi\nabla V
\right).
\end{align}
Expanding the divergence gives
\begin{align}
\nabla\cdot(\rho\nabla V)+\Delta\rho
&=
\nabla\phi\cdot\nabla\psi
+
\phi\Delta\psi
+
\frac12\nabla\phi\cdot(\psi\nabla V)
+
\frac12\phi\nabla\psi\cdot\nabla V
+
\frac12\phi\psi\Delta V .
\end{align}
The terms involving $\nabla\psi\cdot\nabla V$ cancel because $\nabla\phi=-\frac12\phi\nabla V$. Hence
\begin{align}
\nabla\cdot(\rho\nabla V)+\Delta\rho
&=
\phi\Delta\psi
-
\frac14\phi\psi|\nabla V|^2
+
\frac12\phi\psi\Delta V .
\end{align}
Dividing by $\phi$ gives
\begin{equation}
\label{eq:app-ground-state-transform-result}
\phi^{-1}\left[\nabla\cdot(\rho\nabla V)+\Delta\rho\right]
=
\Delta\psi
-
\frac14|\nabla V|^2\psi
+
\frac12\Delta V\,\psi
=
-H_V\psi,
\end{equation}
where
\begin{equation}
\label{eq:app-schrodinger-H}
H_V
=
-\Delta
+
\frac14|\nabla V|^2
-
\frac12\Delta V .
\end{equation}

Now consider the residual drift. We have
\begin{align}
\phi^{-1}\nabla\cdot(\rho G)
&=
\phi^{-1}\nabla\cdot(\phi\psi G) \nonumber\\
&=
G\cdot\nabla\psi
+
\psi\,\phi^{-1}G\cdot\nabla\phi
+
\psi\,\nabla\cdot G .
\end{align}
Using $\phi^{-1}G\cdot\nabla\phi=-\frac12G\cdot\nabla V$ and the weighted divergence-free condition
$$
\nabla\cdot(e^{-V}G)=0
\quad\Longleftrightarrow\quad
\nabla\cdot G=G\cdot\nabla V,
$$
we obtain
\begin{equation}
\label{eq:app-skew-transport-G}
\phi^{-1}\nabla\cdot(\rho G)
=
K_G\psi,
\qquad
K_G\psi
=
G\cdot\nabla\psi+
\frac12(\nabla\cdot G)\psi .
\end{equation}
Combining the two contributions proves
\begin{equation}
\partial_t\psi=-H_V\psi+K_G\psi.
\end{equation}

The operator $K_G$ is skew-adjoint in $L^2(dx)$. Indeed, for smooth compactly supported functions $f,g$,
\begin{align}
\langle f,K_G g\rangle
&=
\int f\,G\cdot\nabla g\,dx
+
\frac12\int f g\,\nabla\cdot G\,dx \\ 
&=
-\int g\,G\cdot\nabla f\,dx
-\frac12\int f g\,\nabla\cdot G\,dx \\ 
&=
-\langle K_G f,g\rangle,
\end{align}
where we used integration by parts. Thus, writing $H_G:=-iK_G$, the operator $H_G$ is formally self-adjoint and $K_G=+iH_G$.

We also verify that the stationary state is annihilated by both pieces. For the reversible part, define
$$
D_j
=
\partial_{x_j}
+
\frac12\partial_{x_j}V.
$$
The formal adjoint in $L^2(\mathbb R^N,dx)$ is
$$
D_j^\dagger
=
-\partial_{x_j}
+
\frac12\partial_{x_j}V.
$$
Then
\begin{align}
D_j^\dagger D_j f
&=
\left(
-\partial_{x_j}
+
\frac12\partial_{x_j}V
\right)
\left(
\partial_{x_j}f
+
\frac12(\partial_{x_j}V)f
\right) \nonumber\\
&=
-\partial_{x_j}^2 f
-
\frac12(\partial_{x_j}^2V)f
+
\frac14(\partial_{x_j}V)^2f .
\end{align}
Summing over $j$ gives
$$
\sum_{j=1}^N D_j^\dagger D_j
=
-\Delta
+
\frac14|\nabla V|^2
-
\frac12\Delta V
=
H_V.
$$
Thus $H_V\geq 0$ on its natural quadratic-form domain. Since
$$
D_j\phi
=
\partial_{x_j}(e^{-V/2})
+
\frac12(\partial_{x_j}V)e^{-V/2}
=
0,
$$
we have
$$
H_V\phi=\sum_{j=1}^N D_j^\dagger D_j\phi=0.
$$
For the residual part,
$$
K_G\phi
=
G\cdot\nabla\phi+
\frac12(\nabla\cdot G)\phi
=
-\frac12(G\cdot\nabla V)\phi+
\frac12(\nabla\cdot G)\phi
=
0,
$$
again using $\nabla\cdot G=G\cdot\nabla V$. Hence $\phi=\sqrt{\rho_\infty}$ is a zero mode of the full transformed generator $-H_V+K_G$. Setting $G=0$ gives the reversible ground-state transform used in the posterior Hamiltonian construction.

\subsection{Quadratic confinement and oscillator form}
\label{app:quadratic-confinement-oscillator-form}

In this subsection we specialize to the reversible case and write $H:=H_V$. 
We derive \Cref{eq:quadratic-potential-number-operator,eq:H-number-plus-potential}. If
$$
V(x)=\|x\|^2,
$$
then
$$
\nabla V=2x,
\qquad
|\nabla V|^2=4\|x\|^2,
\qquad
\Delta V=2N.
$$
Substituting in \Cref{eq:app-schrodinger-H} gives
$$
H
=
-\Delta+\|x\|^2-N.
$$
Using the oscillator identity
$$
2\hat n=-\Delta+\|\hat q\|^2-N,
$$
we obtain
$$
H=2\hat n.
$$

Now let
$$
V(x)=\|x\|^2+F_0(x).
$$
Then
$$
\nabla V=2x+\nabla F_0,
\qquad
\Delta V=2N+\Delta F_0.
$$
Therefore
\begin{align}
H
&=
-\Delta
+
\frac14|2x+\nabla F_0(x)|^2
-
\frac12(2N+\Delta F_0(x)) \nonumber\\
&=
-\Delta+\|x\|^2-N
+
x\cdot\nabla F_0(x)
+
\frac14|\nabla F_0(x)|^2
-
\frac12\Delta F_0(x).
\end{align}
The first three terms are $2\hat n$. Hence
$$
H=2\hat n+W_{F_0}(\hat q),
$$
where
$$
W_{F_0}(x)
=
x\cdot\nabla F_0(x)
+
\frac14|\nabla F_0(x)|^2
-
\frac12\Delta F_0(x).
$$
If $F_0$ is a polynomial, then $W_{F_0}$ is also a polynomial. In particular, the transformed operator is of the form covered by the polynomial block-encoding results of \Cref{subsec:block-encoding-polynomial-operators}.

\subsection{Bayesian inverse problems in whitened coordinates}
\label{app:bayesian-inverse-problems-whitening}

We derive \Cref{eq:inverse-problem-potential}. The observation model is
$$
y=g(u)+\eta,
\qquad
\eta\sim \mathcal N(0,\Gamma),
$$
and the prior is
$$
u\sim \mathcal N(\mu,\Sigma).
$$
The posterior density with respect to Lebesgue measure on $u$ is proportional to
\begin{equation}
\label{eq:app-posterior-u}
\exp\left(
-\frac12\|u-\mu\|_\Sigma^2
-\frac12\|y-g(u)\|_\Gamma^2
\right),
\end{equation}
where
$$
\|u-\mu\|_\Sigma^2=(u-\mu)^T\Sigma^{-1}(u-\mu).
$$

Introduce the whitened coordinate
$$
u=\mu+\sqrt 2\,\Sigma^{1/2}x.
$$
Then
$$
\frac12\|u-\mu\|_\Sigma^2
=
\frac12
\left\|
\sqrt 2\,x
\right\|^2
=
\|x\|^2.
$$
The likelihood term becomes
$$
\frac12
\left\|
y-g(\mu+\sqrt 2\,\Sigma^{1/2}x)
\right\|_\Gamma^2.
$$
The Jacobian of the affine change of variables is constant and is absorbed into the normalizing constant. Therefore the posterior density in $x$ is
$$
\pi(x|y)
=
Z_y^{-1}e^{-V_y(x)},
$$
with
$$
V_y(x)
=
\|x\|^2
+
\frac12
\left\|
y-g(\mu+\sqrt 2\,\Sigma^{1/2}x)
\right\|_\Gamma^2.
$$
This proves \Cref{eq:inverse-problem-potential}.

Finally, for any observable $f$ for which the integral is well defined,
$$
\mathbb E_{\pi(\cdot|y)}[f(x)]
=
\int f(x)\pi(x|y)\,dx.
$$
Since $\phi_y=\sqrt{\pi(\cdot|y)}$, this is
$$
\int f(x)|\phi_y(x)|^2dx
=
\langle \phi_y,f(\hat q)\phi_y\rangle.
$$
This proves \Cref{eq:posterior-expectation-ground-state}.

\subsection{Structured forward maps and posterior Hamiltonian block encodings}
\label{app:structured-forward-map-posterior-hamiltonian}

\Cref{subsec:asymptotics-bayes} is stated directly in terms of $W_y(\hat q)$ having an efficient structured operator-polynomial representation. Here we give a useful sufficient condition in terms of the whitened residual
$$
r_y(x)=\Gamma^{-1/2}\left(y-g(\mu+\sqrt2\,\Sigma^{1/2}x)\right),
$$
so that
\begin{equation}
\label{eq:app-Fy-residual}
F_y(x)
=
\frac12\|r_y(x)\|^2
=
\frac12\sum_{a=1}^M r_{y,a}(x)^2.
\end{equation}
The posterior Hamiltonian is
\begin{equation}
\label{eq:app-Hy-Wy}
H_y
=
2\hat n+W_y(\hat q),
\qquad
W_y(x)
=
x\cdot\nabla F_y(x)
+
\frac14|\nabla F_y(x)|^2
-
\frac12\Delta F_y(x).
\end{equation}

Assume that $M=\operatorname{poly}(n)$ and that each residual component has a constant-degree structured scalar representation
\begin{equation}
\label{eq:app-residual-structured}
r_{y,a}(x)
=
\sum_{\rho=1}^{R_a}
 c_{a,\rho}
\prod_{s=1}^{k_{a,\rho}}
\ell_{a,\rho,s}(x),
\qquad
\ell_{a,\rho,s}(x)=a_{a,\rho,s}+v_{a,\rho,s}\cdot x,
\end{equation}
with $k_{a,\rho}\leq p=O(1)$, $\sum_a R_a=\operatorname{poly}(n)$, polynomially bounded scalar coefficients and vector norms, efficiently preparable normalized coefficient states, and classically computable inner products $v\cdot v'$ between all coefficient vectors appearing in the representation.

We show that these assumptions imply that $W_y(\hat q)$ has an efficient structured operator-polynomial representation with only position operators. It is enough to check closure under the operations appearing in \Cref{eq:app-Hy-Wy}. Let
\begin{equation}
\label{eq:app-structured-monomial}
T(x)
=
 c\prod_{s=1}^k \ell_s(x),
\qquad
\ell_s(x)=a_s+v_s\cdot x,
\qquad
k=O(1).
\end{equation}
Then
\begin{equation}
\label{eq:app-gradient-structured-monomial}
\nabla T(x)
=
c
\sum_{t=1}^k
v_t
\prod_{s\neq t}
\ell_s(x).
\end{equation}
Consequently,
\begin{equation}
\label{eq:app-x-dot-gradient-structured}
x\cdot\nabla T(x)
=
c
\sum_{t=1}^k
(v_t\cdot x)
\prod_{s\neq t}
\ell_s(x),
\end{equation}
which is again a sum of products of affine linear forms.

The Laplacian is also closed in this representation. Since each $\ell_s$ is affine,
\begin{equation}
\label{eq:app-laplacian-structured-monomial}
\Delta T(x)
=
c
\sum_{\substack{t,u=1\\ t\neq u}}^k
(v_t\cdot v_u)
\prod_{s\neq t,u}
\ell_s(x).
\end{equation}
Thus $\Delta T$ is structured provided the inner products $v_t\cdot v_u$ are classically computable.

Now suppose each residual component has the form
\begin{equation}
\label{eq:app-residual-structured-one}
r_{y,a}(x)
=
\sum_{\rho=1}^{R_a}
c_{a,\rho}
\prod_{s=1}^{k_{a,\rho}}
\ell_{a,\rho,s}(x),
\qquad
k_{a,\rho}\leq p.
\end{equation}
Then $F_y=\frac12\sum_a r_{y,a}^2$ is a structured polynomial of degree at most $2p$ and size polynomial in $M$ and the $R_a$'s. Applying \Cref{eq:app-x-dot-gradient-structured,eq:app-laplacian-structured-monomial} term by term shows that
$$
x\cdot\nabla F_y
\qquad\text{and}\qquad
\Delta F_y
$$
also have polynomial-size structured representations.

It remains to check $|\nabla F_y|^2$. By the calculation above, $\nabla F_y$ can be written as
\begin{equation}
\label{eq:app-gradient-Fy-structured}
\nabla F_y(x)
=
\sum_{\rho}
b_\rho(x)v_\rho,
\end{equation}
where each $b_\rho(x)$ is a structured polynomial of degree at most $2p-1$, and each $v_\rho$ is one of the coefficient vectors appearing in the affine forms. Therefore
\begin{equation}
\label{eq:app-gradient-square-structured}
|\nabla F_y(x)|^2
=
\sum_{\rho,\sigma}
(v_\rho\cdot v_\sigma)b_\rho(x)b_\sigma(x).
\end{equation}
This is again a polynomial-size structured polynomial, assuming the inner products $v_\rho\cdot v_\sigma$ are computable. Its degree is at most $4p-2$.

Combining the three pieces in \Cref{eq:app-Hy-Wy}, we conclude that $W_y$ is a structured polynomial of degree at most $4p-2$ and polynomial size. Hence, by \Cref{thm:structured-polynomial-block-encoding}, the truncated operator
$$
\Pi_m W_y(\hat q)\Pi_m
$$
admits a block encoding with polynomial normalization and gate complexity whenever $m=\operatorname{poly}(n)$. The term $2\hat n$ is diagonal in the global Hermite basis and has normalization $O(m)$. Therefore
$$
H_{y,m}
=
\Pi_m(2\hat n+W_y(\hat q))\Pi_m
$$
also admits a block encoding with polynomial normalization and gate complexity.

\subsection{Linear-Gaussian inverse problems}
\label{app:linear-gaussian-classical-reduction}

Consider a linear forward map, the whitened inverse problem has
\begin{equation}
\label{eq:app-linear-gaussian-model}
V(x)
=
\|x\|^2
+
\frac12\|b-Ax\|^2,
\end{equation}
with $A\in\mathbb R^{M\times N}$ and $b\in\mathbb R^M$. Expanding,
\begin{equation}
\label{eq:app-linear-gaussian-expanded}
V(x)
=
\frac12 x^TKx-h^Tx+\frac12\|b\|^2,
\qquad
K=2I+A^TA,
\qquad
h=A^Tb.
\end{equation}
Thus the posterior is Gaussian with precision $K$ and mean $K^{-1}h$. Since
$$
K\succeq 2I,
$$
the potential is uniformly convex.

The transformed Hamiltonian is
\begin{equation}
\label{eq:app-linear-gaussian-H}
H
=
-\Delta
+
\frac14\|Kx-h\|^2
-
\frac12\operatorname{tr}K.
\end{equation}
After shifting $x$ by the posterior mean $K^{-1}h$, this becomes
\begin{equation}
\label{eq:app-linear-gaussian-shifted-H}
H
=
-\Delta
+
\frac14 x^TK^2x
-
\frac12\operatorname{tr}K.
\end{equation}
Diagonalizing $K=U^T\operatorname{diag}(\kappa_1,\ldots,\kappa_N)U$, one obtains a sum of independent harmonic oscillators:
\begin{equation}
\label{eq:app-linear-gaussian-spectrum}
H
=
\sum_{j=1}^N
\left(
-\partial_{z_j}^2
+
\frac{\kappa_j^2}{4}z_j^2
-
\frac{\kappa_j}{2}
\right).
\end{equation}
The spectrum is
\begin{equation}
\label{eq:app-linear-gaussian-eigenvalues}
\left\{
\sum_{j=1}^N n_j\kappa_j:\ n_j\in\mathbb N
\right\}.
\end{equation}
Therefore the spectral gap is
\begin{equation}
\label{eq:app-linear-gaussian-gap}
\gamma
=
\lambda_{\min}(K)
\geq 2.
\end{equation}
The precise constant depends only on the whitening convention; the important point is that the gap is bounded below by a constant because of the Gaussian prior.

The same structure makes many readout tasks classically tractable. By the Woodbury identity,
\begin{equation}
\label{eq:app-woodbury-linear-gaussian}
K^{-1}
=
(2I+A^TA)^{-1}
=
\frac12 I
-
\frac14 A^T
\left(
I+\frac12AA^T
\right)^{-1}
A.
\end{equation}
Thus covariance queries reduce to linear algebra in the $M\times M$ matrix $AA^T$. For example, for any efficiently accessible vector $v\in\mathbb R^N$,
\begin{equation}
\label{eq:app-linear-gaussian-cov-query}
v^TK^{-1}v
=
\frac12\|v\|^2
-
\frac14(Av)^T
\left(
I+\frac12AA^T
\right)^{-1}
(Av).
\end{equation}
Similarly, a coordinate marginal variance is
\begin{equation}
\label{eq:app-linear-gaussian-coordinate-variance}
e_j^TK^{-1}e_j
=
\frac12
-
\frac14(Ae_j)^T
\left(
I+\frac12AA^T
\right)^{-1}
(Ae_j).
\end{equation}
Thus, when $M=\operatorname{poly}(n)$ and the relevant matrix-vector products are classically accessible, the linear-Gaussian case is efficiently reducible to classical linear algebra.

\subsection{Spectral gap for quadratic-square potentials}
\label{app:quadratic-square-spectral-gap}

We prove the spectral-gap statement used in \Cref{subsec:spectral-gap-hardness}. Let
\begin{equation}
\label{eq:app-quadratic-square-potential}
V(x)=\|x\|^2+f(x)^2,
\qquad
f(x)=x^TAx+b\cdot x+c,
\end{equation}
with $A=A^T$. Let
\begin{equation}
\label{eq:app-quadratic-square-H}
H
=
-\Delta
+
\frac14|\nabla V|^2
-
\frac12\Delta V,
\qquad
\phi=Z^{-1/2}e^{-V/2}.
\end{equation}
Then $H\phi=0$ and $H\geq 0$ by the ground-state transform. We show that $0$ is isolated.

Let $d\mu=Z^{-1}e^{-V}dx$. 
For any $\chi=h\phi$, the ground-state transform gives the Dirichlet identity
\begin{equation}
\label{eq:app-dirichlet-identity}
\langle \chi,H\chi\rangle_{L^2(dx)}
=
\int_{\mathbb R^N}|\nabla h|^2\,d\mu.
\end{equation}
Also,
\begin{equation}
\label{eq:app-orthogonality-mean-zero}
\chi\perp \phi
\quad\Longleftrightarrow\quad
\int h\,d\mu=0.
\end{equation}
Thus a Poincaré inequality for $\mu$ implies a spectral gap for $H$.

We now indicate why $\mu$ satisfies such an inequality. Consider the reversible Langevin generator
\begin{equation}
\label{eq:app-langevin-generator}
\mathcal L
=
\Delta-\nabla V\cdot\nabla.
\end{equation}
For sufficiently small $\alpha>0$, set
\begin{equation}
\label{eq:app-lyapunov-Valpha}
V_\alpha(x)=\|x\|^2+\alpha f(x)^2.
\end{equation}
A polynomial coercivity estimate gives constants $c,C>0$ such that
\begin{equation}
\label{eq:app-coercivity-quadratic-square}
\nabla V\cdot\nabla V_\alpha-\Delta V_\alpha
\geq
c\left(
\|x\|^2+f(x)^2|\nabla f(x)|^2
\right)
-C.
\end{equation}
This estimate follows by writing $x=r\theta$ and comparing the leading powers of $r$ along each direction $\theta\in S^{N-1}$. The positive terms either grow like $r^6$, $r^4$, or are controlled by the confining term $\|x\|^2$.

Let
\begin{equation}
\label{eq:app-lyapunov-function}
W(x)=e^{\varepsilon V_\alpha(x)}
\end{equation}
with $\varepsilon>0$ sufficiently small. Then
\begin{equation}
\label{eq:app-lyapunov-drift}
\frac{\mathcal LW}{W}
=
\varepsilon\Delta V_\alpha
+
\varepsilon^2|\nabla V_\alpha|^2
-
\varepsilon\nabla V\cdot\nabla V_\alpha.
\end{equation}
Using \Cref{eq:app-coercivity-quadratic-square} and choosing $\varepsilon$ small enough yields a Lyapunov drift condition:
\begin{equation}
\label{eq:app-lyapunov-condition}
\mathcal LW
\leq
-\kappa W+B\mathbf 1_{B_R}
\end{equation}
for some $\kappa,B,R>0$. Since $e^{-V}$ is smooth and bounded above and below on $B_R$, the measure $\mu$ satisfies a local Poincaré inequality on $B_R$. 
The Lyapunov criterion for Poincaré inequalities then implies a global Poincaré inequality:
\begin{equation}
\label{eq:app-global-poincare}
\int
\left(
h-\int h\,d\mu
\right)^2
d\mu
\leq
C_P
\int |\nabla h|^2\,d\mu.
\end{equation}
See~\cite{bakry2008simple-a43, cattiaux2013poincar-8bd}.

Now take $\chi\perp\phi$ and write $\chi=h\phi$. By \Cref{eq:app-orthogonality-mean-zero}, $\int h\,d\mu=0$. Combining \Cref{eq:app-dirichlet-identity,eq:app-global-poincare} gives
\begin{equation}
\label{eq:app-gap-from-poincare}
\|\chi\|_{L^2(dx)}^2
=
\int h^2\,d\mu
\leq
C_P
\int|\nabla h|^2\,d\mu
=
C_P
\langle\chi,H\chi\rangle.
\end{equation}
Hence
\begin{equation}
\label{eq:app-quadratic-square-gap}
\langle\chi,H\chi\rangle
\geq
C_P^{-1}\|\chi\|^2,
\qquad
\chi\perp\phi.
\end{equation}
Therefore
\begin{equation}
\label{eq:app-quadratic-square-spectrum}
\sigma(H)\subset\{0\}\cup[C_P^{-1},\infty).
\end{equation}
Finally, if $\chi\in\ker H$ and $\chi=h\phi$, then \Cref{eq:app-dirichlet-identity} gives $\nabla h=0$ $\mu$-almost everywhere. Since $\mu$ is equivalent to Lebesgue measure on the connected space $\mathbb R^N$, $h$ is constant. Thus
\begin{equation}
\label{eq:app-quadratic-square-kernel}
\ker H=\operatorname{span}\{\phi\}.
\end{equation}
This proves existence of a positive gap for each fixed instance. The argument does not give an inverse-polynomial lower bound on $C_P^{-1}$ as the dimension or the instance parameters vary.

\subsection{Quartic multiwell potentials and Ising hardness}
\label{app:ising-hardness-quartic-potentials}

We explain the hardness statement used in \Cref{subsec:spectral-gap-hardness}. The point is not that every quartic inverse problem is hard. Rather, the class of quartic multiwell potentials is expressive enough to approximate discrete spin systems.

Let $G$ be an Ising instance with spins $s_i\in\{\pm1\}$ and energy
\begin{equation}
\label{eq:app-ising-energy}
E_G(s)
=
-\sum_{(i,j)}J_{ij}s_is_j
-
\sum_i h_i s_i.
\end{equation}
Consider the quartic potential
\begin{equation}
\label{eq:app-ising-quartic-potential}
V_{G,\kappa}(x)
=
\|x\|^2
+
\kappa\sum_{i=1}^n(x_i^2-1)^2
+
\beta E_G(x),
\end{equation}
where $E_G(x)$ is the same polynomial with $s_i$ replaced by $x_i$. The quadratic prior term $\|x\|^2$ is constant on the hypercube vertices and therefore does not change the relative spin weights.

For large $\kappa$, the term $\kappa\sum_i(x_i^2-1)^2$ creates $2^n$ wells near the vertices $s\in\{\pm1\}^n$. Laplace approximation gives
\begin{equation}
\label{eq:app-ising-partition-approx}
\int_{\mathbb R^n}e^{-V_{G,\kappa}(x)}dx
=
C_{\kappa,n}e^{-n}
\sum_{s\in\{\pm1\}^n}e^{-\beta E_G(s)}
\left(1+o_{\kappa}(1)\right),
\end{equation}
where the factor $C_{\kappa,n}$ is independent of the spin configuration when the wells are symmetric. Likewise, for observables $f$ that vary slowly within each well,
\begin{equation}
\label{eq:app-ising-observable-approx}
\frac{\int f(x)e^{-V_{G,\kappa}(x)}dx}
{\int e^{-V_{G,\kappa}(x)}dx}
=
\frac{\sum_s f(s)e^{-\beta E_G(s)}}
{\sum_s e^{-\beta E_G(s)}}
+o_{\kappa}(1).
\end{equation}
Therefore ground-state observables of the continuous Hamiltonian approximate Ising Gibbs expectations, because
$$
\phi_{G,\kappa}(x)=Z^{-1/2}e^{-V_{G,\kappa}(x)/2}
$$
satisfies
\begin{equation}
\label{eq:app-ground-state-gibbs-expectation}
\langle \phi_{G,\kappa},f(\hat q)\phi_{G,\kappa}\rangle
=
\frac{\int f(x)e^{-V_{G,\kappa}(x)}dx}
{\int e^{-V_{G,\kappa}(x)}dx}.
\end{equation}

This construction is compatible with the inverse-problem form. The pinning wells arise from quadratic residuals $g_i(x)=x_i^2$ with observations near $1$, giving terms $(x_i^2-1)^2$. Pair interactions arise from quadratic residuals $g_{ij}(x)=x_ix_j$ with nonzero observations: expanding $(x_ix_j-y_{ij})^2$ gives a term $-2y_{ij}x_ix_j$ plus quartic terms that are approximately constant inside the wells. Linear fields are produced similarly from affine residuals. Thus quadratic forward maps can generate quartic multiwell posterior landscapes whose low-temperature Gibbs weights approximate an Ising model.

Exact evaluation of Ising and related Tutte partition functions is $\#$P-hard in broad parameter regimes, and approximation with fields can also be hard depending on the parameter regime~\cite{jaeger1990computational-6cf,goldberg2006complexity-d2d}. Consequently, one should not expect a general polynomial-time algorithm for arbitrary quartic posterior landscapes. This hardness statement is compatible with the spectral-gap existence result in \Cref{app:quadratic-square-spectral-gap}: the existence theorem gives $\gamma>0$ for fixed instances, but does not provide an inverse-polynomial lower bound. In multiwell regimes, the gap can be exponentially small because of metastability between well-separated explanations of the data.

\subsection{Uniform spectral gap from a Hessian lower bound}
\label{app:hessian-lower-bound-gap}

We prove the spectral-gap statement used in \Cref{subsec:spectral-gap-hardness}.
Let
$$
V_\tau(x):=\|x\|^2+\tau F(x),
\qquad
0\le \tau\le 1,
$$
and define
\begin{equation}
H_\tau
:=
-\Delta
+
\frac14|\nabla V_\tau|^2
-
\frac12\Delta V_\tau .
\label{eq:interpolating-witten-hamiltonian}
\end{equation}
Assume that there exists $\beta<2$ such that
\begin{equation}
\nabla^2F(x)\succeq -\beta I
\qquad
\text{for all }x\in\mathbb R^N .
\label{eq:hessian-lower-bound-F}
\end{equation}
Set
$$
\kappa:=2-\beta>0.
$$
Then, for all $\tau\in[0,1]$,
\begin{equation}
\nabla^2V_\tau(x)
=
2I+\tau\nabla^2F(x)
\succeq
(2-\beta\tau)I
\succeq
\kappa I .
\label{eq:uniform-convexity-Vtau}
\end{equation}
First, this implies that $V_\tau$ is uniformly confining. Indeed, fixing $x\in\mathbb R^N$ and applying Taylor's formula to
$g(t):=V_\tau(tx)$ gives
$$
V_\tau(x)
\ge
V_\tau(0)+\nabla V_\tau(0)\cdot x+\frac{\kappa}{2}\|x\|^2.
$$
Thus
$$
Z_\tau:=\int_{\mathbb R^N}e^{-V_\tau(x)}\,dx<+\infty.
$$
Define
$$
d\mu_\tau(x):=Z_\tau^{-1}e^{-V_\tau(x)}\,dx,
\qquad
\phi_\tau:=Z_\tau^{-1/2}e^{-V_\tau/2}.
$$
We now record the ground-state transform identity. Let
$$
D_j(\tau):=\partial_{x_j}+\frac12\partial_{x_j}V_\tau .
$$
Then
$$
D_j(\tau)\phi_\tau=0,
$$
and
\begin{equation}
H_\tau=\sum_{j=1}^N D_j(\tau)^\dagger D_j(\tau).
\label{eq:Htau-factorization}
\end{equation}
Consequently, for any smooth $h$ we define $\chi:=h\phi_\tau$,
$$
D_j(\tau)(h\phi_\tau) = (\partial_{x_j}h)\phi_\tau+hD_j(\tau)\phi_\tau = (\partial_{x_j}h)\phi_\tau.
$$
Defining $d\mu_\tau(x) := \phi^2_{\tau}(x) dx$
\begin{equation}
\langle \chi,H_\tau\chi\rangle_{L^2(dx)}
= \sum_{j=1}^N \|D_j(\tau)\chi\|^2
=
\int_{\mathbb R^N}|\nabla h(x)|^2\,d\mu_\tau(x).
\label{eq:dirichlet-identity-hessian-gap}
\end{equation}
Moreover if $\chi$ is orthogonal to $\phi_\tau$ then it has zero mean in the $\mu_\tau$ measure,
\begin{equation}
\chi\perp\phi_\tau
\quad\Longleftrightarrow\quad
 0 =\langle \chi,\phi_\tau\rangle_{L^2(dx)}=
\int_{\mathbb R^N}h\,d\mu_\tau .
\label{eq:orthogonality-zero-mean}
\end{equation}
We use the standard Poincaré inequality for uniformly strongly convex Gibbs measures: if
$\nabla^2V_\tau\succeq \kappa I$, then
\begin{equation}
\int_{\mathbb R^N}
\left|
h-\int h\,d\mu_\tau
\right|^2
d\mu_\tau
\le
\frac1\kappa
\int_{\mathbb R^N}|\nabla h|^2\,d\mu_\tau .
\label{eq:strong-convexity-poincare}
\end{equation}
Applying this to a function $h$ with zero $\mu_\tau$-mean gives
$$
\int |h|^2\,d\mu_\tau
\le
\frac1\kappa
\int |\nabla h|^2\,d\mu_\tau .
$$
Therefore, if $\chi=h\phi_\tau$ and $\chi\perp\phi_\tau$, then
\begin{equation}
\|\chi\|^2
=
\int |h|^2\,d\mu_\tau
\le
\frac1\kappa
\int |\nabla h|^2\,d\mu_\tau
=
\frac1\kappa
\langle \chi,H_\tau\chi\rangle .
\label{eq:gap-rayleigh-hessian}
\end{equation}
Equivalently,
\begin{equation}
\langle \chi,H_\tau\chi\rangle
\ge
\kappa\|\chi\|^2,
\qquad
\chi\perp\phi_\tau .
\label{eq:hessian-gap-rayleigh-final}
\end{equation}
Thus
\begin{equation}
\sigma(H_\tau)\subset \{0\}\cup[\kappa,\infty),
\qquad
\gamma(H_\tau)\ge \kappa=2-\beta .
\label{eq:uniform-hessian-gap-final}
\end{equation}
Finally, the ground state is unique. If $H_\tau\chi=0$ and $\chi=h\phi_\tau$, then
\eqref{eq:dirichlet-identity-hessian-gap} gives
$$
\int |\nabla h|^2\,d\mu_\tau=0.
$$
Hence $h$ is constant $\mu_\tau$-almost everywhere, and since $\mu_\tau$ is equivalent to Lebesgue measure on the connected space $\mathbb R^N$,
$$
\ker H_\tau=\mathrm{span}\{\phi_\tau\}.
$$

\end{document}